\documentclass[aps,pra,twocolumn,superscriptaddress,longbibliography]{revtex4-2}

\usepackage{amsmath,amssymb,amsfonts,mathtools}
\usepackage{amsthm}
\usepackage{physics}
\usepackage{braket}
\usepackage{bbm}
\usepackage{mathrsfs}
\usepackage{xcolor}
\usepackage{graphicx}
\usepackage{subcaption}
\usepackage{multirow}
\usepackage{tabularx}
\usepackage{enumitem}
\usepackage{microtype}
\usepackage{tikz-cd}
\usepackage{proof}
\usepackage{thmtools,thm-restate}
\usepackage{algorithmicx}
\usepackage{algpseudocode}
\usepackage{qcircuit}
\usepackage{float}
\usepackage{soul}
\usepackage{url}
\usepackage[colorlinks,citecolor=blue,linkcolor=blue,urlcolor=blue,hypertexnames=false]{hyperref}

\DeclareMathAlphabet{\mathpzc}{OT1}{pzc}{m}{it}
\DeclareMathOperator{\aff}{aff}
\DeclareMathOperator{\STAB}{STAB}

\newtheorem{theorem}{Theorem}[section]

\newtheorem{corollary}[theorem]{Corollary}

\newcounter{prob}
\newtheorem{problem}[prob]{Problem}

\newcommand{\eq}[1]{\hyperref[eq:#1]{(\ref*{eq:#1})}}
\renewcommand{\sec}[1]{\hyperref[sec:#1]{Sec.~\ref*{sec:#1}}}
\newcommand{\app}[1]{\hyperref[app:#1]{Appendix~\ref*{app:#1}}}
\newcommand{\fig}[1]{\hyperref[fig:#1]{Fig.~\ref*{fig:#1}}}
\newcommand{\tbl}[1]{\hyperref[tbl:#1]{Table~\ref*{tbl:#1}}}
\newcommand{\thm}[1]{\hyperref[thm:#1]{Theorem~\ref*{thm:#1}}}
\newcommand{\cor}[1]{\hyperref[cor:#1]{Corollary~\ref*{cor:#1}}}
\newcommand{\prop}[1]{\hyperref[prop:#1]{Proposition~\ref*{prop:#1}}}
\newcommand{\lem}[1]{\hyperref[lem:#1]{Lemma~\ref*{lem:#1}}}
\newcommand{\prob}[1]{\hyperref[prob:#1]{Problem~\ref*{prob:#1}}}

\newcommand{\safeincludegraphics}[2][]{\IfFileExists{#2}{\includegraphics[#1]{#2}}{\fbox{\parbox[c][0.18\textheight][c]{0.42\linewidth}{\centering Missing figure:\\#2}}}}

\begin{document}

\title{Affine-Profile Stabilizer Thresholds for Magic in Codeword-Stabilized Quantum Codes}

\author{Li-Yi Hsu}
\affiliation{Department of Physics, Chung Yuan Christian University, Chungli 32081, Taiwan}
\affiliation{Physics Division, National Center for Theoretical Sciences, Taipei 106319, Taiwan}

\author{En-Jui Kuo}
\affiliation{Department of Electrophysics, National Yang Ming Chiao Tung University, Hsinchu, Taiwan, R.O.C.}

\begin{abstract}
Codeword-stabilized quantum codes give a unified graph-state description of stabilizer and nonadditive quantum error-correcting codes.  Although each individual CWS word state is stabilizer, coherent superpositions of different word states can be nonstabilizer.  We develop a CWS-adapted magic-witness framework that isolates this codeword coherence and converts it into certified lower bounds on robustness of magic.

The main result is an exact reduction of the stabilizer threshold of a natural CWS coherence witness to a finite-geometric problem over the classical CWS word set.  For general weighted superpositions, the threshold is computed by enumerating affine intersections and affine-quadratic phases.  For equal-weight superpositions, the phase optimization collapses, and the threshold is determined entirely by how many CWS words can lie in affine flats of each dimension.  Thus a quantum optimization over stabilizer states becomes a classical affine-incidence problem.

This reduction yields a fixed-parameter algorithm, an analytic lower bound for an infinite union-stabilizer family, and exact rational certificates for several standard nonadditive CWS examples.  The framework also clarifies why exact enumeration fails for large structured families and identifies the remaining task as an affine-intersection problem.  The result provides a geometric mechanism by which nonlinear CWS word sets generate certifiable magic.
\end{abstract}

\maketitle


\section{Introduction}
\label{sec:introduction}

Entanglement~\cite{Horodecki2009}, coherence~\cite{Streltsov2017}, Bell nonlocality~\cite{Brunner2014}, and nonstabilizerness or magic~\cite{PhysRevLett.118.090501,Veitch2014,hamaguchi2024handbook} are distinct quantum resources.  Magic is central in fault-tolerant quantum computation: stabilizer states and Clifford operations are efficiently simulable by the Gottesman--Knill theorem~\cite{gottesman1998heisenberg,aaronson2004improved}, whereas non-Clifford resources are required for universal quantum computation.  Magic and stabilizer structure have also appeared in randomness generation~\cite{VairogsYan2025}, communication complexity~\cite{ChowdhuryEtAl2025}, many-body dynamics and monitored circuits~\cite{LiuWinter2022,Russomanno2025,TurkeshiDymarskySierant2025,BeraSchiro2025,Hartse2025,KorbanyGullansPiroli2025,Niroula2024,Bejan2024,Fux2024,Suzuki2025,SierantTurkeshi2026,Turkeshi2025Spreading,Leone2021Chaos,Varikuti2026}, symplectic and phase-space methods~\cite{gross2006hudson,hostens2005stabilizer,HeinrichGross2019,LeoneOlivieroHamma2022,Macedo2025,mele2026symplectic,crew2025magic,hughes2014quantum}, and experimental or numerical diagnostics~\cite{chabaud2020stellar,hahn2026assessing,weinbub2018recent,harriman1988some,takahashi1986wigner}.

The robustness of magic (RoM) \cite{PhysRevLett.118.090501} quantifies the distance of a state from the convex hull of stabilizer states.  Exact RoM is a linear program over all stabilizer states, and therefore becomes expensive rapidly.  A common way to obtain lower bounds is to use a Hermitian witness: if an observable has a large expectation on a target state but a bounded value on all stabilizer states, it certifies nonstabilizerness.  Bell operators fit this general pattern after fixing local measurement observables, although the threshold relevant for magic is the stabilizer threshold, not the local-hidden-variable threshold~\cite{brunner2012testing,Baccari2020,bowles2018device,mckague2011self,kuo2024self}.  We keep this Bell-inspired viewpoint as motivation, but the main objects in this paper are Pauli witnesses adapted to CWS codes.

Recent works have also emphasized that tractable magic certification often requires exploiting additional combinatorial structure.  One approach starts from a limited Pauli measurement set and studies the corresponding reduced stabilizer polytope; there the relevant object is the frustration graph of the measured Pauli observables, together with sign-dependency constraints from the stabilizer formalism~\cite{liu2026graph}.  A separate and very recent work studies nonstabilizerness of CWS code states through stabilizer R\'enyi entropy and additive combinatorics, where the relevant classical data are translation-collision statistics, difference multiplicities, and related energy quantities of the CWS word set~\cite{LiuXu2026CWSMagic}.  Our work is complementary to both.  We construct a Pauli-measurable CWS codeword-coherence witness and compute the stabilizer threshold entering a robustness-of-magic lower bound.  In the equal-weight case, this threshold is governed by affine-flat intersections of the CWS word set, rather than by frustration graphs, collision counts, or parallelogram-type additive energies.

A CWS code in standard form is specified by a graph state and a classical word set.  Its word states are Pauli translates of the graph state, hence each individual word state is stabilizer.  This makes CWS codes a natural setting in which to ask where magic can arise.  It cannot arise from a single word state, and it is not certified merely by membership in the CWS code space.  We show that the relevant resource is the coherent interference among different word states.  To isolate this resource, we define a CWS codeword-coherence witness by subtracting the diagonal codeword population from the projector onto a coherent CWS superposition.

The main technical result is a CWS-specific simplification of the stabilizer-threshold problem.  In the graph basis, every stabilizer state has affine support and affine-quadratic phase.  Hence the expectation value of the CWS coherence witness depends only on which CWS words are contained in an affine flat, together with the phase induced on that intersection.  For general weighted superpositions this gives an exact fixed-parameter algorithm in the number of CWS words.  For equal-weight superpositions the phase optimization disappears, and the stabilizer threshold is exactly the affine-intersection profile of the classical word set.

This result gives a concrete geometric interpretation.  The quantum optimization over all stabilizer states is converted into the following finite-geometric question: how many CWS words can lie in an affine flat of a given dimension?  This affine profile determines the exact stabilizer threshold for equal-weight CWS coherence witnesses.  We apply the formula to an infinite binary union-stabilizer family, where it gives an analytic RoM lower bound growing linearly with block length, and to several standard nonadditive CWS examples, where it yields exact rational certificates.

Our contributions are fourfold.  First, we formulate a witness-induced lower bound on robustness of magic that applies to any calibrated Hermitian observable with a known stabilizer threshold.  Second, we introduce a CWS codeword-coherence witness for weighted superpositions of CWS word states.  Third, we prove that the stabilizer threshold of this witness admits an exact fixed-parameter algorithm in the number of CWS words, and that in the equal-weight case it collapses to the affine-intersection profile of the classical word set.  Fourth, we apply the formula to QEC families and standard nonadditive CWS examples, obtaining analytic and exact rational certificates.

The paper is organized as follows.  Section~\ref{sec:preliminaries} introduces stabilizer states, graph states, CWS codes, and affine geometry over binary vector spaces.  Section~\ref{sec:rom-bound} gives the witness-induced RoM lower bound and explains the Bell-type motivation.  Section~\ref{sec:cws-general} defines general weighted CWS coherence witnesses.  Section~\ref{sec:equal-weight} proves the equal-weight affine-profile formula.  Section~\ref{sec:qec-families} applies the method to QEC families and standard nonadditive CWS examples.  The appendices contain the RoM proof and complexity accounting, the graph-basis derivation, Bell-type verification examples and figures, CWS data used in the computations, and NISQ measurement and noise discussions.

\section{Preliminaries}
\label{sec:preliminaries}

\subsection{Stabilizer states and stabilizer polytope}
\label{subsec:prelim-stabilizer}

Let \(\mathcal P_n\) be the \(n\)-qubit Pauli group, with global phases ignored when no confusion can arise.  A pure stabilizer state is the unique common \(+1\)-eigenstate of a maximal abelian subgroup \(S\le \mathcal P_n\) generated by \(n\) independent commuting Pauli operators and not containing \(-I\).  We denote the convex hull of pure stabilizer states by
\begin{align}
    \STAB_n
    :=
    \operatorname{conv}\{\ket{\phi}\bra{\phi}:\ket{\phi}\text{ is a }n\text{-qubit stabilizer state}\}.
    \label{eq:stab-polytope-def}
\end{align}
The Clifford group is the normalizer of the Pauli group, and Clifford unitaries map stabilizer states to stabilizer states.  This observation is used below when moving between the computational basis and the graph basis.
We use the affine-quadratic normal form for stabilizer states~\cite{dehaene2003clifford,gross2007lu}. If \(\ket{\eta}\) is any stabilizer state, then there exist an affine support \(V+t\subseteq\mathbb F_2^n\), with \(V\le\mathbb F_2^n\) an \(r\)-dimensional linear subspace, and quadratic functions \(q\) over $\mathbb{F}_2^n$ and linear function \(\ell\) over $\mathbb{F}_2^n$ such that
\begin{align}
    \ket{\eta}
    =
    2^{-r/2}
    \sum_{x\in V+t}
    i^{\ell(x)}(-1)^{q(x)}\ket{x}.
    \label{eq:stab-affine-quadratic-prelim}
\end{align}
The precise normal form is reviewed in Appendix~\ref{app:graph-basis-proof} in the graph-basis setting used for CWS codes.

\subsection{Graph states and graph basis}
\label{subsec:prelim-graph-states}

Let \(G=(V,E)\) be a simple graph on \(n\) vertices.  The associated graph state is
\begin{align}
    \ket{G}
    =
    \prod_{\{i,j\}\in E}\mathrm{CZ}_{ij}\ket{+}^{\otimes n}.
    \label{eq:graph-state-def}
\end{align}
Equivalently, \(\ket{G}\) is the unique common \(+1\)-eigenstate of the graph-state stabilizer generators
\begin{align}
    K_j
    =
    X_j\prod_{k\in N(j)}Z_k,
    \qquad
    j=1,\ldots,n,
    \label{eq:graph-stabilizer-generator}
\end{align}
where \(N(j)\) is the neighborhood of vertex \(j\).  For \(u=(u_1,\ldots,u_n)\in\mathbb F_2^n\), define
\begin{align}
    Z^u
    :=
    Z_1^{u_1}\cdots Z_n^{u_n},
    \qquad
    \ket{u_G}:=Z^u\ket{G}.
    \label{eq:graph-basis-def}
\end{align}
The states \(\{\ket{u_G}:u\in\mathbb F_2^n\}\) form an orthonormal graph basis.  If
\begin{align}
    U_G
    :=
    \left(\prod_{\{i,j\}\in E}\mathrm{CZ}_{ij}\right)H^{\otimes n},
    \label{eq:graph-basis-clifford}
\end{align}
then
\begin{align}
    U_G\ket{u}=\ket{u_G}.
    \label{eq:graph-basis-unitary-action}
\end{align}
Since \(U_G\) is Clifford, transforming a stabilizer state by \(U_G\) or \(U_G^\dagger\) preserves stabilizer structure.

\subsection{CWS codes in standard form}
\label{subsec:prelim-cws}

A codeword-stabilized (CWS) code in standard form is specified by a graph state \(\ket{G}\) and a classical word set \(\mathcal C\subseteq\mathbb F_2^n\).  The CWS code space is
\begin{align}
    \mathcal Q(G,\mathcal C)
    =
    \operatorname{span}\{\ket{c_G}:c\in\mathcal C\},
    \qquad
    \ket{c_G}=Z^c\ket{G}.
    \label{eq:cws-code-space}
\end{align}
The code projector is
\begin{align}
    \Pi_{\mathcal C}
    =
    \sum_{c\in\mathcal C}\ket{c_G}\bra{c_G}.
    \label{eq:cws-code-projector}
\end{align}
The CWS formalism unifies stabilizer and nonadditive quantum codes~\cite{Rains1997,Cross2009,Chuang2009,YenHsu2009}.  If \(K=|\mathcal C|\) is not a power of two, then the CWS code is necessarily nonadditive.  Every individual word state \(\ket{c_G}\) is stabilizer, because it is obtained from \(\ket{G}\) by a Pauli operator.  The resource considered here is therefore not the magic of individual word states, but the coherence among them.

For normalized amplitudes \(\alpha=(\alpha_c)_{c\in\mathcal C}\), consider the weighted CWS superposition
\begin{align}
    \ket{\Psi_\alpha}
    =
    \sum_{c\in\mathcal C}\alpha_c\ket{c_G},
    \qquad
    \sum_{c\in\mathcal C}|\alpha_c|^2=1.
    \label{eq:weighted-cws-state-prelim}
\end{align}
The equal-weight case is
\begin{align}
    \ket{\Psi_{\mathcal C}}
    =
    \frac{1}{\sqrt K}
    \sum_{c\in\mathcal C}\ket{c_G},
    \qquad
    K=|\mathcal C|.
    \label{eq:equal-cws-state-prelim}
\end{align}

We finally introduce the finite-geometric quantity that will control the equal-weight CWS threshold.
\subsection{Affine flats and affine profiles}
\label{subsec:prelim-affine}

An affine flat in \(\mathbb F_2^n\) is a translate \(A=t+V\), where \(V\le\mathbb F_2^n\) is a linear subspace.  Its dimension is \(\dim A:=\dim V\).  For a nonempty subset \(D=\{d_1,\ldots,d_m\}\subseteq\mathbb F_2^n\), the smallest affine flat containing \(D\), denoted by \(\aff(D)\), is
\begin{align}
    \aff(D)
    =
    d_1+
    \operatorname{span}_{\mathbb F_2}\{d_2+d_1,\ldots,d_m+d_1\}.
    \label{eq:affine-hull-def}
\end{align}
It is unique and satisfies
\begin{align}
    \dim\aff(D)
    \le
    |D|-1.
    \label{eq:affine-hull-dimension-bound}
\end{align}
For a CWS word set \(\mathcal C\), define its affine profile by
\begin{align}
    M_r(\mathcal C)
    :=
    \max_{\substack{A\subseteq\mathbb F_2^n\\ A\;{\rm affine}\\ \dim A=r}}
    |A\cap\mathcal C|.
    \label{eq:affine-profile-def}
\end{align}
This finite-geometric quantity will determine the equal-weight stabilizer threshold in Theorem~\ref{thm:equal-weight-affine-profile-main}.  It obeys the elementary bounds
\begin{align}
    \max\!\left(1,\left\lceil |\mathcal C|2^{r-n}\right\rceil\right)
    \le
    M_r(\mathcal C)
    \le
    \min(|\mathcal C|,2^r),
    \label{eq:trivial-bounds}
\end{align}
where the lower bound follows by averaging over affine \(r\)-flats and the upper bound is immediate from the size of an affine \(r\)-flat.


\section{Robustness of magic and witness lower bounds}
\label{sec:rom-bound}

We first recall a simple witness principle for lower-bounding robustness of magic.  The statement is elementary, but it is useful because it separates the certification problem into two parts: choosing an observable with a large target expectation value and computing its stabilizer threshold.

The robustness of magic of an \(n\)-qubit state \(\rho\) is \cite{webster2022universal}
\begin{align}
    R(\rho)
    :=
    \min_{\sigma_+,\sigma_-\in\STAB_n}
    \left\{2p+1\;\middle|\;
    \rho=(p+1)\sigma_+-p\sigma_-,\;p\ge0\right\}.
    \label{eq:rom-def-main}
\end{align}
Equivalently, if \(\{\sigma_j\}\) denotes the set of pure stabilizer states, then
\begin{align}
    R(\rho)
    =
    \min_x
    \left\{
    \|x\|_1\;\middle|\;
    \rho=\sum_j x_j\sigma_j
    \right\}.
    \label{eq:rom-l1-main}
\end{align}
This form makes clear why exact RoM is difficult: the number of stabilizer states is exponential in \(n^2\).  The following theorem turns any calibrated Hermitian observable with a known stabilizer threshold into a RoM lower bound, in the same spirit as witness-based approaches to nonstabilizerness~\cite{Macedo2025}.

\begin{theorem}[Witness-induced RoM lower bound]
\label{thm:witness-rom-bound}
Let \(B\) be a Hermitian operator and define
\begin{align}
    \beta_B(\rho)
    :=
    \Tr(B\rho),
    \qquad
    \beta_{\rm stab}(B)
    :=
    \max_{\sigma\in\STAB_n}|\Tr(B\sigma)|.
    \label{eq:witness-beta-defs}
\end{align}
Then
\begin{align}
    R(\rho)
    \ge
    \frac{|\beta_B(\rho)|}{\beta_{\rm stab}(B)}.
    \label{eq:witness-rom-lower-bound-main}
\end{align}
In particular, if \(|\Tr(B\rho)|>\beta_{\rm stab}(B)\), then \(\rho\) is nonstabilizer and the bound is nontrivial.
\end{theorem}

The proof is given in Appendix~\ref{app:rom-proof-and-complexity}.  A Bell operator may be used as \(B\) after its measurement observables are fixed.  In that case, the denominator in Eq.~\eqref{eq:witness-rom-lower-bound-main} is the stabilizer threshold, not the local-hidden-variable threshold.  We record this observation because graph-state Bell operators are typically written as combinations of stabilizer generators~\cite{Baccari2020,bowles2018device,kuo2024self}, and because the original version of this project was motivated by such Bell-type observables.

\begin{corollary}[Bell-type observables as RoM witnesses]
\label{cor:bell-type-rom-witness}
Let \(\mathcal B\) be a Bell-type Hermitian operator after the local measurement observables are fixed.  Then
\begin{align}
    R(\rho)
    \ge
    \frac{|\Tr(\mathcal B\rho)|}{\max_{\sigma\in\STAB_n}|\Tr(\mathcal B\sigma)|}.
    \label{eq:bell-type-rom-bound-main}
\end{align}
This is a measurement-dependent magic witness.  It is not a device-independent Bell certificate unless the denominator is replaced by a local-hidden-variable threshold.
\end{corollary}

The preceding statements are deliberately general: they say how any calibrated
Hermitian observable becomes a magic witness once its stabilizer threshold is
known.  The rest of the paper is about making this denominator computable for a
class of observables that is not arbitrary.  We now use the CWS structure to
construct a coherence witness whose threshold depends only on the affine
geometry of the classical word set.

To illustrate the general witness mechanism before specializing to CWS codes, Fig.~\ref{fig:combined-bell-examples} shows small Bell-type examples in which a calibrated Hermitian observable gives a RoM lower bound after its stabilizer threshold is computed.  The model details are given in Appendix~\ref{subsec:bell-checks}; the CWS affine-profile theorem below is independent of these numerical scans.

\begin{figure*}[t]
\centering
\begin{subfigure}{0.32\textwidth}
    \centering
    \includegraphics[width=\linewidth]{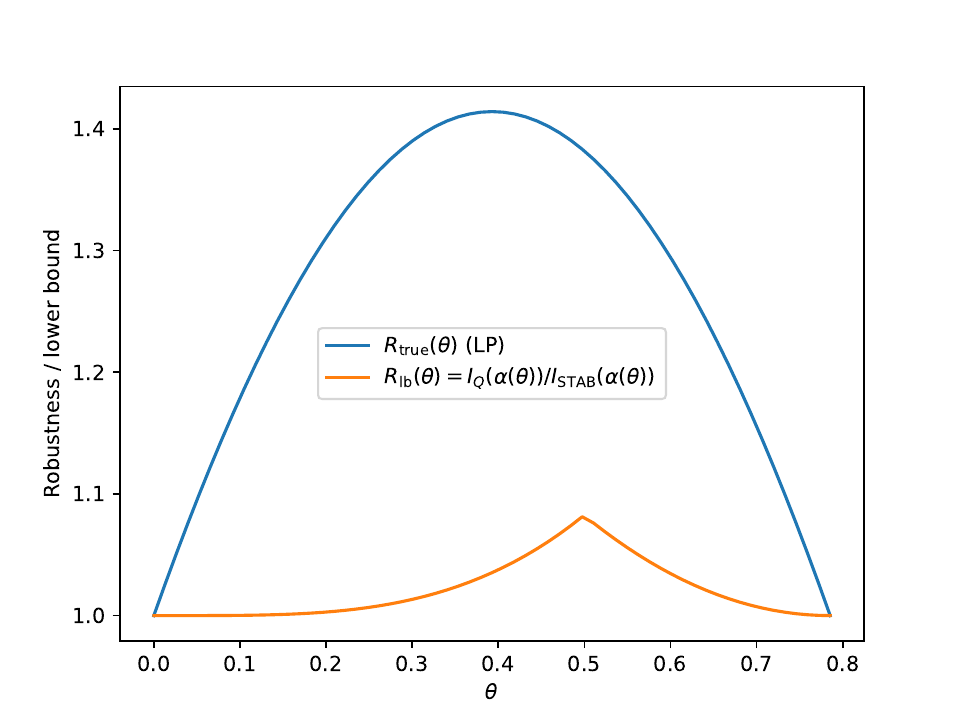}
    \caption{Two-qubit tilted-CHSH scan.}
    \label{fig:2-qubit}
\end{subfigure}\hfill
\begin{subfigure}{0.32\textwidth}
    \centering
    \includegraphics[width=\linewidth]{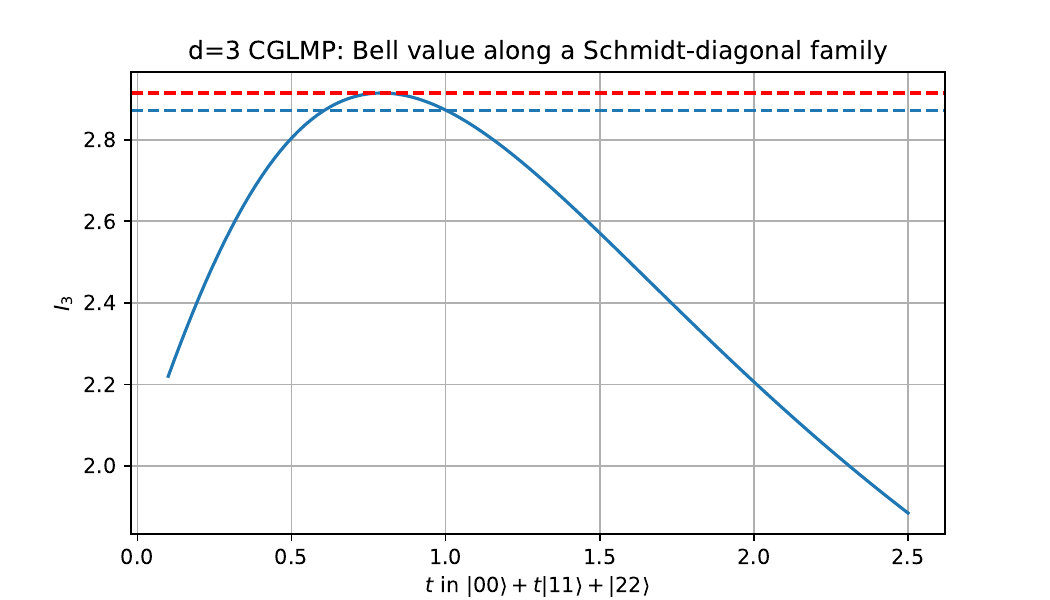}
    \caption{Generalized-W line scan.}
    \label{fig:3-qubit}
\end{subfigure}\hfill
\begin{subfigure}{0.32\textwidth}
    \centering
    \includegraphics[width=\linewidth]{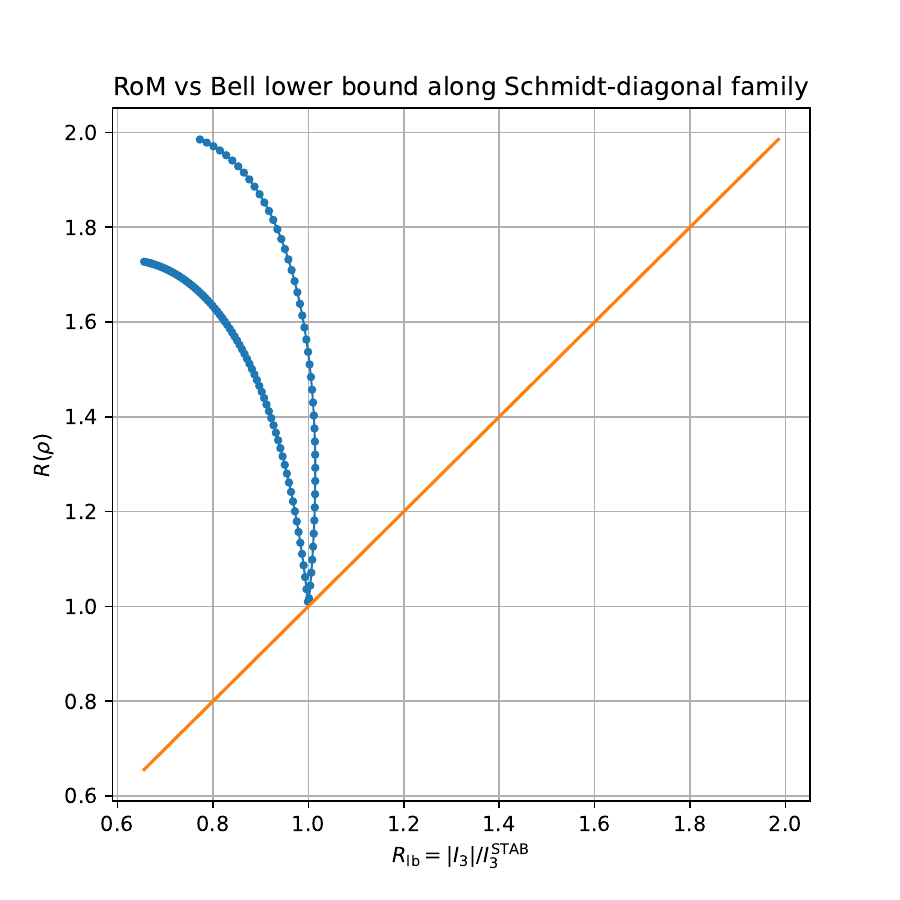}
    \caption{Three-qubit lower-bound scatter plot.}
    \label{fig:3-qubit-2}
\end{subfigure}
\caption{
Small-scale numerical checks of the witness-induced RoM lower-bound
principle.  Panel (a) shows the two-qubit family
\(\ket{\psi_\theta}=\cos\theta\ket{00}+\sin\theta\ket{11}\): the blue curve
is the exact RoM computed by the stabilizer-decomposition linear program,
while the orange curve is the tilted-CHSH stabilizer-threshold lower bound.
Panel (b) shows the symmetric line \(\phi=\pi/4\) in the generalized
three-qubit W family.  Again, the blue curve is the exact RoM and the orange
curve is the optimized Bell-type lower bound \(B^\star/6\); the vertical
dashed line marks the standard \(\ket W\) state.  Panel (c) plots the same
three-qubit scan as exact RoM versus Bell-certified lower bound.  The
diagonal line is \(y=x\), and the points lying on or above it verify that
the Bell-type quantity is a valid RoM lower bound.  These examples
illustrate Theorem~\ref{thm:witness-rom-bound}; they are independent of the
CWS affine-profile theorem developed later.  The tilted-CHSH
and three-qubit Bell functionals used in the plots are cited and specified
in Appendix~\ref{subsec:bell-checks}.
}
\label{fig:combined-bell-examples}
\end{figure*}

\section{General weighted CWS coherence witnesses}
\label{sec:cws-general}

The witness theorem leaves open the choice of the observable \(B\).  For a CWS code, there is a canonical choice when the target state is a coherent superposition of word states.  Since each word state is itself stabilizer, diagonal codeword population is not the resource we want to certify.  We therefore subtract the diagonal part and retain only off-diagonal codeword coherence.

Let \((G,\mathcal C)\) be a CWS code and let \(\ket{\Psi_\alpha}\) be the weighted superposition in Eq.~\eqref{eq:weighted-cws-state-prelim}.  The diagonal codeword population associated with the weights is
\begin{align}
    D_\alpha
    :=
    \sum_{c\in\mathcal C}|\alpha_c|^2\ket{c_G}\bra{c_G}.
    \label{eq:weighted-diagonal-population}
\end{align}
We define the CWS codeword-coherence witness by subtracting this diagonal contribution:
\begin{align}
    B_{\rm coh}^{(\alpha)}
    &:={}
    \ket{\Psi_\alpha}\bra{\Psi_\alpha}-D_\alpha
    \notag\\
    &=
    \sum_{\substack{c,c'\in\mathcal C\\c\ne c'}}
    \alpha_c\alpha_{c'}^*\ket{c_G}\bra{c'_G}.
    \label{eq:weighted-cws-coherence-witness}
\end{align}
It vanishes on every incoherent mixture of CWS word states:
\begin{align}
    \Tr\!\left[
    B_{\rm coh}^{(\alpha)}
    \sum_{c\in\mathcal C}p_c\ket{c_G}\bra{c_G}
    \right]
    =0.
    \label{eq:coherence-witness-vanishes-mixtures}
\end{align}
Thus \(B_{\rm coh}^{(\alpha)}\) probes codeword coherence rather than code-space population.

The numerator in the witness ratio is analytic.  Orthogonality of the graph-basis states gives
\begin{align}
    \bra{\Psi_\alpha}B_{\rm coh}^{(\alpha)}\ket{\Psi_\alpha}
    &=
    1-
    \sum_{c\in\mathcal C}|
    \alpha_c|^2|\braket{c_G|\Psi_\alpha}|^2
    \notag\\
    &=
    1-
    \sum_{c\in\mathcal C}|\alpha_c|^4.
    \label{eq:weighted-target-numerator-main}
\end{align}
The nontrivial part is the denominator, which requires an optimization over all \(n\)-qubit stabilizer states.  A naive enumeration involves \(2^{\Theta(n^2)}\) stabilizer states: 
\begin{align}
    \beta_{\rm coh}^{(\alpha)}
    :=
    \max_{\ket{\phi}\in\STAB_n}
    \left|
    \bra{\phi}B_{\rm coh}^{(\alpha)}\ket{\phi}
    \right|.
    \label{eq:weighted-coherence-denominator-main}
\end{align}
The special structure of CWS codes makes this optimization much smaller.  The following theorem is the general weighted version of the algorithmic result.

\begin{theorem}[Weighted CWS coherence algorithm]
\label{thm:weighted-cws-algorithm-main}
For a CWS word set \(\mathcal C\subseteq\mathbb F_2^n\) with \(K=|\mathcal C|\), the threshold \(\beta_{\rm coh}^{(\alpha)}\) can be computed exactly in time
\begin{align}
    T_{\rm wt}(n,K)
    =
    O\!\left(
    2^K(n^3+Kn^2)
    +
    K2^{\binom K2+3K}
    \right).
    \label{eq:weighted-runtime-main}
\end{align}
Consequently,
\begin{align}
    R(\Psi_\alpha)
    \ge
    \frac{1-\sum_{c\in\mathcal C}|\alpha_c|^4}{\beta_{\rm coh}^{(\alpha)}}.
    \label{eq:weighted-cws-rom-lower-bound-main}
\end{align}
\end{theorem}

The detailed proof is given in Appendix~\ref{app:graph-basis-proof}.  The key point is that in the graph basis every stabilizer state has affine support and affine-quadratic phase.  Hence the threshold in Eq.~\eqref{eq:weighted-coherence-denominator-main} is computed by enumerating intersections \(D=\mathcal C\cap A\), where \(A\) is an affine flat, and by optimizing the allowed phase on \(D\).  Since \(\dim\aff(D)\le |D|-1\le K-1\), the exponential part depends on \(K\), not on the total number of \(n\)-qubit stabilizer states.

The witness is also Pauli-measurable.  If \(S_G\) is the graph-state stabilizer group, then
\begin{align}
    \ket{G}\bra{G}
    =
    2^{-n}\sum_{s\in S_G}s.
    \label{eq:graph-projector-pauli-main}
\end{align}
Consequently,
\begin{align}
    \ket{c_G}\bra{c'_G}
    =
    2^{-n}\sum_{s\in S_G}Z^c s Z^{c'}.
    \label{eq:transition-pauli-expansion-main}
\end{align}
Equations~\eqref{eq:weighted-cws-coherence-witness} and~\eqref{eq:transition-pauli-expansion-main} give an analytic Pauli expansion of \(B_{\rm coh}^{(\alpha)}\), determined entirely by \((G,\mathcal C,\alpha)\).

The weighted statement is useful for general CWS superpositions, but the equal-weight case in Eq.~\eqref{eq:equal-cws-state-prelim} is especially natural for code families.  In that case the phase optimization appearing in the weighted algorithm has a closed-form solution, and the stabilizer threshold becomes a function only of the affine profile of the word set \(\mathcal C\).


\section{Equal-weight affine-profile formula}
\label{sec:equal-weight}

The weighted theorem is useful as an algorithmic statement, but the main geometric simplification occurs for equal-weight CWS superpositions.  In this case the phase optimization disappears: for any fixed affine support, the constant phase is allowed and maximizes the coherent sum.  The entire stabilizer threshold is therefore determined only by the affine-intersection profile of the classical word set.

We now specialize to the equal-weight state in Eq.~\eqref{eq:equal-cws-state-prelim}.  The witness becomes
\begin{align}
    B_{\rm coh}
    =
    \ket{\Psi_{\mathcal C}}\bra{\Psi_{\mathcal C}}
    -
    \frac1K\sum_{c\in\mathcal C}\ket{c_G}\bra{c_G}.
    \label{eq:equal-coherence-witness-main}
\end{align}
For equal weights, the phase optimization in Theorem~\ref{thm:weighted-cws-algorithm-main} collapses.  The reason is simple: for a fixed affine support \(A\), the best phase alignment is achieved by the constant phase, which is a valid affine-quadratic phase. We present it using the following Theorem.

\begin{theorem}[Equal-weight affine-profile threshold]
\label{thm:equal-weight-affine-profile-main}
For the equal-weight CWS coherence witness,
\begin{align}
    \beta_{\rm coh}
    &:=
    \max_{\ket{\phi}\in\STAB_n}
    \left|
    \bra{\phi}B_{\rm coh}\ket{\phi}
    \right|
    \notag\\
    &=
    \max_{0\le r\le n}
    \frac{M_r(\mathcal C)\bigl(M_r(\mathcal C)-1\bigr)}{K2^r}.
    \label{eq:equal-affine-profile-formula-main}
\end{align}
Consequently,
\begin{align}
    R(\Psi_{\mathcal C})
    \ge
    \frac{1-\frac1K}{\beta_{\rm coh}}.
    \label{eq:equal-weight-rom-main}
\end{align}
Moreover, \(\beta_{\rm coh}\) can be computed exactly by subset enumeration in time
\begin{align}
    T_{\rm eq}(n,K)
    =
    O\!\left(2^K(n^3+Kn^2)\right).
    \label{eq:equal-subset-runtime-main}
\end{align}
\end{theorem}

The proof is given in Appendix~\ref{app:graph-basis-proof}.  This theorem is the geometric core of the paper: the quantum stabilizer threshold is exactly the affine-intersection profile of the classical word set.  For small \(K\), subset enumeration is efficient.  For small \(n\) but larger \(K\), it may be better to enumerate affine flats directly.  Since the number of affine \(r\)-flats in \(\mathbb F_2^n\) is \(2^{n-r}{n\brack r}_2\), where \({n\brack r}_2\) is the binary Gaussian binomial coefficient, this gives the alternative exact runtime
\begin{align}
    T_{\rm flat}(n,K)
    =
    O\!\left(
    K\sum_{r=0}^{n}2^{n-r}{n\brack r}_2
    \right)
    =
    2^{n^2/4+O(n)}K.
    \label{eq:affine-flat-runtime-main}
\end{align}
In the examples below, one may use whichever enumeration is smaller.

A useful structural consequence is that linear word sets do not generate magic in the equal-weight mode.  If \(\mathcal C\le\mathbb F_2^n\) is a linear subspace, then character orthogonality implies that \(\ket{\Psi_{\mathcal C}}\) is a uniform quadratic-phase state on a linear subspace, hence a stabilizer state.  Nontrivial equal-weight magic therefore arises from nonlinear affine geometry of \(\mathcal C\).


\section{QEC families and examples}
\label{sec:qec-families}

We now demonstrate that the affine-profile formula is not merely a formal reduction.  It gives an analytic bound for an infinite QEC family and exact rational certificates for several standard nonadditive CWS codes.  In all cases, the certified state is the equal coherent superposition over the listed CWS word states; we do not claim that every state in the corresponding code space has magic.

\subsection{The binary \texorpdfstring{\(((n,n+1,2))\)}{((n,n+1,2))} union-stabilizer family}
\label{subsec:ust-family}

We begin with the binary member of the \(((n,1+n(q-1),2))_q\) union-stabilizer family~\cite{eczoo_arvind}.  In the qubit case \(q=2\), the parameters are \(((n,n+1,2))\).  The case \(n=5\) recovers the cyclic \(((5,6,2))\) code, which detects single-qubit errors~\cite{eczoo_qubit_5_6_2}.  Up to an affine change of coordinates, the CWS word set can be chosen as
\begin{align}
    \mathcal C_n
    =
    \{0,e_1,\ldots,e_n\}
    \subseteq\mathbb F_2^n,
    \qquad
    K=n+1.
    \label{eq:affine-simplex-word-set-main}
\end{align}

\begin{theorem}[Affine-simplex union-stabilizer family]
\label{thm:affine-simplex-family-main}
For \(n\ge2\), the equal-weight CWS superposition over \(\mathcal C_n\) satisfies
\begin{align}
    \beta_{\rm coh}
    =
    \frac{3}{2(n+1)},
    \label{eq:affine-simplex-beta-main}
\end{align}
and therefore
\begin{align}
    R(\Psi_{\mathcal C_n})
    \ge
    \frac{2n}{3}.
    \label{eq:affine-simplex-rom-main}
\end{align}
\end{theorem}

\begin{proof}
The \(n+1\) points of \(\mathcal C_n\) are affinely independent.  Hence an \(r\)-dimensional affine flat contains at most \(r+1\) of them, and this is tight by taking the flat spanned by \(0,e_1,\ldots,e_r\).  Thus
\begin{align}
    M_r(\mathcal C_n)=r+1.
    \label{eq:affine-simplex-profile-main}
\end{align}
Using Theorem~\ref{thm:equal-weight-affine-profile-main},
\begin{align}
    \beta_{\rm coh}
    =
    \max_r\frac{r(r+1)}{(n+1)2^r}.
    \label{eq:affine-simplex-beta-optimization-main}
\end{align}
The function \(r(r+1)/2^r\) is maximized at \(r=2\) and \(r=3\), where it equals \(3/2\).  This gives Eq.~\eqref{eq:affine-simplex-beta-main}.  The numerator is \(1-1/(n+1)=n/(n+1)\), which gives Eq.~\eqref{eq:affine-simplex-rom-main}.  For \(((5,6,2))\), this gives the lower bound \(10/3\), matching the first row of Table~\ref{tab:cws-high-distance-examples}.
\end{proof}

\subsection{High-distance and structured nonadditive CWS examples}
\label{subsec:high-distance-cws-examples}

We next apply the affine-profile algorithm to standard nonadditive CWS examples.  These include the nine-qubit cyclic CWS code \(((9,12,3))\), a nonadditive code correcting arbitrary single-qubit errors~\cite{eczoo_qubit_9_12_3}, and the \(((10,18,3))\) and \(((10,20,3))\) codes found in the CWS framework~\cite{Cross2009}.  The last row gives the first nontrivial SSW distance-two example.  It is included as a structured case where direct affine-flat enumeration is more convenient than subset enumeration.  Since the relevant values of \(K\) are between \(6\) and \(22\), the exact affine-profile computation is feasible.

For each equal-weight CWS state, we compute Eq.~\eqref{eq:equal-affine-profile-formula-main}.  The maximizing affine flat is summarized by the pair
\begin{align}
    (r_*,m_*),
    \qquad
    m_*:=M_{r_*}(\mathcal C),
    \label{eq:best-affine-pair-main}
\end{align}
so that
\begin{align}
    \beta_{\rm coh}
    =
    \frac{m_*(m_*-1)}{K2^{r_*}}.
    \label{eq:best-affine-pair-beta-main}
\end{align}

\begin{table*}[t]
\centering
\begin{tabular}{c|c|c|c|c|c}
Code & \(K\) & \(d\) & \((r_*,m_*)\) & \(\beta_{\rm coh}\) &
\(R(\Psi_{\mathcal C})\ge (1-1/K)/\beta_{\rm coh}\) \\
\hline
\(((5,6,2))\) & \(6\) & \(2\) & \((2,3)\) & \(1/4\) & \(10/3\) \\
\(((9,12,3))\) & \(12\) & \(3\) & \((3,6)\) & \(5/16\) & \(44/15\) \\
\(((10,18,3))\) & \(18\) & \(3\) & \((3,8)\) & \(7/18\) & \(17/7\) \\
\(((10,20,3))\) & \(20\) & \(3\) & \((4,12)\) & \(33/80\) & \(76/33\) \\
\(((7,22,2))_{\rm SSW}\) & \(22\) & \(2\) & \((5,16)\) & \(15/44\) & \(14/5\)
\end{tabular}
\caption{Exact affine-profile magic lower bounds for standard nonadditive CWS examples.  The pair \((r_*,m_*)\) records an affine dimension \(r_*\) and an intersection size \(m_*=M_{r_*}(\mathcal C)\) attaining the stabilizer threshold.  All entries are exact rational values obtained from the affine-profile computation.}
\label{tab:cws-high-distance-examples}
\end{table*}

For instance, the \(((9,12,3))\) row is obtained from an affine 3-flat containing six CWS words:
\begin{align}
    \beta_{\rm coh}
    =
    \frac{6\cdot5}{12\cdot2^3}
    =
    \frac{5}{16},
    \qquad
    R\ge
    \frac{11/12}{5/16}
    =
    \frac{44}{15}.
    \label{eq:912-example-calculation-main}
\end{align}
The rational form is a consequence of the equal-weight affine-profile formula.  The CWS word sets and affine-profile data used to compute the table are recorded in Appendix~\ref{subsec:cws-data-computations}.
\subsection{The SSW distance-two family}
\label{subsec:ssw-family}

The SSW distance-two nonadditive family also admits a CWS description in Ref.~\cite{Cross2009}.  For odd \(n\), the original codewords have the form \(\ket{x}+\ket{\bar x}\), with \(x\) chosen according to a weight-parity rule.  If \(n\equiv1\pmod4\), one takes odd-weight strings satisfying \(|x|<(n-1)/2\); if \(n\equiv3\pmod4\), one takes even-weight strings with the same bound.

In the CWS formulation, the stabilizer is equivalent to a star-graph stabilizer and the word operators may be written as
\begin{align}
    W_x=X_1^{x_1}Z_2^{x_2}\cdots Z_n^{x_n}.
    \label{eq:ssw-word-operators-main}
\end{align}
Passing to standard CWS form gives the word label
\begin{align}
    c(x)
    =
    (0,x_2+x_1,x_3+x_1,\ldots,x_n+x_1)
    \in\mathbb F_2^n.
    \label{eq:ssw-standard-label-main}
\end{align}
The pair \(x,\bar x\) gives the same label.  For \(n=7\), this produces \(K=22\) CWS words.  Direct affine-flat enumeration gives
\begin{align}
    (M_0,M_1,M_2,M_3,M_4,M_5,M_6,M_7)
    =
    (1,2,4,7,11,16,22,22),
    \label{eq:ssw-seven-profile-main}
\end{align}
which yields the last row of Table~\ref{tab:cws-high-distance-examples}.

\subsection{The Rains family and analytic affine profiles}
\label{subsec:rains-family}

The Rains family has parameters
\begin{align}
    \left(\left(2m+1,\,3\times2^{2m-3},\,2\right)\right),
    \qquad
    m\ge2,
    \label{eq:rains-family-parameters-main}
\end{align}
and includes the \(((5,6,2))\) code as the case \(m=2\)~\cite{Rains1997,eczoo_rains,rains1997nonadditive,rains1999quantum}.  Here
\begin{align}
    n=2m+1,
    \qquad
    K=3\times2^{2m-3}=3\times2^{n-4}.
    \label{eq:rains-family-k-main}
\end{align}
Thus \(K\) is exponential in \(n\), and exact subset enumeration is not scalable.  The equal-weight formula nevertheless gives a useful reduction:
\begin{align}
    \beta_{\rm coh}
    =
    \max_r
    \frac{M_r(\mathcal C_{\rm Rains})(M_r(\mathcal C_{\rm Rains})-1)}{K2^r}.
    \label{eq:rains-affine-profile-main}
\end{align}
Therefore, any analytic upper bound
\begin{align}
    M_r(\mathcal C_{\rm Rains})\le U_r
    \label{eq:rains-profile-upper-bound-main}
\end{align}
immediately gives
\begin{align}
    R(\Psi_{\mathcal C_{\rm Rains}})
    \ge
    \frac{1-\frac1K}{\displaystyle\max_r\frac{U_r(U_r-1)}{K2^r}}.
    \label{eq:rains-rom-from-profile-bound-main}
\end{align}
Thus, for large structured CWS families, the remaining challenge is no longer stabilizer-state enumeration but bounding the affine-intersection profile of the associated classical word set.  The Rains family therefore marks the boundary of the present exact enumeration method.  Our theorem does not by itself give a scalable algorithm for this family; instead, it identifies the precise finite-geometric quantity whose upper bound would imply scalable RoM lower bounds.
\begin{problem}[Affine profile of the Rains word set]
\label{prob:rains-affine-profile}
Find useful upper bounds on \(M_r(\mathcal C_{\rm Rains})\).  Any such bound immediately gives scalable RoM lower bounds for equal-weight Rains CWS superpositions through Eq.~\eqref{eq:rains-rom-from-profile-bound-main}.
\end{problem}


\section{Conclusion}
\label{sec:conclusion}

We developed CWS codeword-coherence witnesses for lower-bounding robustness of magic.  The construction is based on a simple resource-theoretic observation: CWS word states are stabilizer states individually, so the nontrivial resource is not code-space membership but coherent interference among different word states.  The witness \(B_{\rm coh}^{(\alpha)}\) removes diagonal codeword population and isolates this interference.

The main technical result is the affine-profile reduction.  For general weighted CWS superpositions, the stabilizer threshold is computed by affine intersections and affine-quadratic phases.  For equal-weight superpositions, the phase optimization collapses and the threshold is governed exactly by the affine-intersection profile \(M_r(\mathcal C)\).  This converts a quantum stabilizer-threshold problem into a classical finite-geometric incidence problem over the CWS word set.

The formula gives an analytic family of lower bounds for the binary union-stabilizer distance-two family and exact rational certificates for several standard nonadditive CWS codes.  It also identifies the obstacle for large families such as the Rains family: one needs analytic upper bounds on affine intersections, rather than brute-force stabilizer enumeration.  Future directions include proving such bounds for structured CWS families and searching for coherence-constrained Bell inequalities supported on the Pauli expansions of CWS coherence witnesses.

\section*{Acknowledgements}

L.-Y.H. acknowledges financial support from the National Science and Technology Council (NSTC) of Taiwan under Grant No.~NSTC~115-2112-M-033-009.  E.-J.K. acknowledges financial support from the National Science and Technology Council (NSTC) of Taiwan under Grant No.~NSTC~114-2112-M-A49-036-MY3.
\section*{Code availability}

The numerical code used in this work has two components.  One component computes the affine-intersection profiles and stabilizer-threshold values in Table~\ref{tab:cws-high-distance-examples} using linear algebra over \(\mathbb F_2\).  The other component reproduces the Bell-type checks in Fig.~\ref{fig:combined-bell-examples}, including the tilted-CHSH scan, the generalized-W scan, and the exact RoM linear programs for \(n\le3\).  The code needed to reproduce the numerical entries and plots will be made publicly available in a GitHub repository upon publication, and is available from the corresponding author upon reasonable request.


\appendix

\section{Proof of the RoM lower bound and complexity accounting}
\label{app:rom-proof-and-complexity}

\subsection{Proof of the witness-induced RoM lower bound}
\label{subsec:appendix-rom-proof}

In this appendix, we prove Theorem~\ref{thm:witness-rom-bound}.

\begin{proof}[Proof of Theorem~\ref{thm:witness-rom-bound}]
By Eq.~\eqref{eq:rom-l1-main}, there exists a real vector \(x\) such that
\begin{align}
    \rho=\sum_j x_j\sigma_j,
    \qquad
    \sum_j|x_j|=R(\rho),
    \label{eq:rom-proof-decomposition}
\end{align}
where the \(\sigma_j\) are pure stabilizer states.  Applying the linear functional defined by \(B\),
\begin{align}
    \Tr(B\rho)
    =
    \sum_j x_j\Tr(B\sigma_j).
    \label{eq:rom-proof-linear-functional}
\end{align}
Therefore
\begin{align}
    |\Tr(B\rho)|
    &\le
    \sum_j |x_j|\,|\Tr(B\sigma_j)|
    \notag\\
    &\le
    \beta_{\rm stab}(B)\sum_j|x_j|
    \notag\\
    &=
    \beta_{\rm stab}(B)R(\rho).
    \label{eq:rom-proof-triangle-bound}
\end{align}
Rearranging gives Eq.~\eqref{eq:witness-rom-lower-bound-main}.
\end{proof}


\section{Graph-basis proof: general weights and equal weights}
\label{app:graph-basis-proof}

\subsection{Stabilizer amplitudes in the graph basis}
\label{subsec:graph-basis-stabilizer-amplitudes}

Let \(\ket{\phi}\) be a stabilizer state.  Since \(U_G\) is Clifford, \(U_G^\dagger\ket{\phi}\) is also a stabilizer state.  Therefore there exist an affine support \(V+t\), with \(\dim V=r\), and phase functions \(q\) and \(\ell\) such that
\begin{align}
    U_G^\dagger\ket{\phi}
    =
    2^{-r/2}
    \sum_{u\in V+t}
    i^{\ell(u)}(-1)^{q(u)}\ket{u}.
    \label{eq:rotated-stabilizer-normal-form}
\end{align}
Writing
\begin{align}
    \chi(u):=i^{\ell(u)}(-1)^{q(u)},
    \label{eq:chi-def-appendix}
\end{align}
and multiplying by \(U_G\), we get
\begin{align}
    \ket{\phi}
    =
    2^{-r/2}
    \sum_{u\in V+t}\chi(u)\ket{u_G}.
    \label{eq:stabilizer-graph-basis-expansion}
\end{align}
Hence
\begin{align}
    \braket{u_G|\phi}
    =
    \begin{cases}
    2^{-r/2}\chi(u), & u\in V+t,\\
    0, & u\notin V+t.
    \end{cases}
    \label{eq:graph-basis-amplitude-appendix}
\end{align}
This is the only stabilizer-state input needed for the CWS reduction.
Let \(D=\mathcal C\cap(V+t)\).  From Eq.~\eqref{eq:graph-basis-amplitude-appendix},
\begin{align}
    \braket{\Psi_\alpha|\phi}
    =
    2^{-r/2}
    \sum_{c\in D}\alpha_c^*\chi(c),
    \label{eq:weighted-psi-phi-overlap}
\end{align}
and
\begin{align}
    \sum_{c\in\mathcal C}|
    \alpha_c|^2|\braket{c_G|\phi}|^2
    =
    2^{-r}
    \sum_{c\in D}|\alpha_c|^2.
    \label{eq:weighted-diagonal-overlap}
\end{align}
Therefore
\begin{align}
    \bra{\phi}B_{\rm coh}^{(\alpha)}\ket{\phi}
    =
    2^{-r}
    \left[
    \left|
    \sum_{c\in D}\alpha_c^*\chi(c)
    \right|^2
    -
    \sum_{c\in D}|\alpha_c|^2
    \right].
    \label{eq:weighted-overlap-formula-appendix}
\end{align}
This proves the overlap formula used in Theorem~\ref{thm:weighted-cws-algorithm-main}.  The target numerator is Eq.~\eqref{eq:weighted-target-numerator-main}.

For equal-weight case \(\alpha_c=1/\sqrt K\), Eq.~\eqref{eq:weighted-overlap-formula-appendix} becomes
\begin{align}
    \bra{\phi}B_{\rm coh}\ket{\phi}
    =
    \frac{2^{-r}}{K}
    \left[
    \left|
    \sum_{c\in D}\chi(c)
    \right|^2
    -
    |D|
    \right].
    \label{eq:equal-overlap-formula-appendix}
\end{align}
If \(m=|D|\), then \(|\sum_{c\in D}\chi(c)|\le m\).  This upper bound is achieved by the constant phase \(\chi(c)=1\), which is a valid affine-quadratic phase.  Thus, for a fixed affine support \(A\), the maximal value is just
\begin{align}
    \frac{|A\cap\mathcal C|(|A\cap\mathcal C|-1)}{K2^{\dim A}}.
    \label{eq:equal-fixed-affine-support-value}
\end{align}
Maximizing over affine supports gives Theorem~\ref{thm:equal-weight-affine-profile-main}.

For completeness, we recall why linear word sets do not generate magic in the equal-weight mode.  In the computational basis,
\begin{align}
    \ket{G}
    =
    2^{-n/2}\sum_{x\in\mathbb F_2^n}(-1)^{q_G(x)}\ket{x}.
    \label{eq:graph-state-computational-quadratic}
\end{align}
Therefore
\begin{align}
    \ket{\Psi_{\mathcal C}}
    =
    2^{-n/2}K^{-1/2}
    \sum_{x\in\mathbb F_2^n}
    (-1)^{q_G(x)}
    \left(\sum_{c\in\mathcal C}(-1)^{c\cdot x}\right)
    \ket{x}.
    \label{eq:cws-walsh-profile}
\end{align}
If \(\mathcal C\le\mathbb F_2^n\) is linear, then character orthogonality gives
\begin{align}
    \sum_{c\in\mathcal C}(-1)^{c\cdot x}
    =
    K\,\mathbf 1_{x\in\mathcal C^\perp}.
    \label{eq:linear-character-orthogonality}
\end{align}
Thus \(\ket{\Psi_{\mathcal C}}\) is a uniform quadratic-phase state supported on a linear subspace, hence a stabilizer state.

\subsection{Time complexity of the weighted algorithm}
\label{subsec:weighted-complexity-proof}

We now spell out the weighted algorithm and its running time.  The input is a
CWS word set
\begin{align}
    \mathcal C
    =
    \{c_1,\ldots,c_K\}
    \subseteq \mathbb F_2^n
    \label{eq:weighted-algorithm-input-code}
\end{align}
together with normalized amplitudes \(\{\alpha_c\}_{c\in\mathcal C}\).  The
goal is to compute the stabilizer threshold
\begin{align}
    \beta_{\rm coh}^{(\alpha)}
    =
    \max_{\ket{\phi}\in\STAB_n}
    \left|
    \bra{\phi}B_{\rm coh}^{(\alpha)}\ket{\phi}
    \right|
    \label{eq:weighted-algorithm-goal-beta}
\end{align}
and hence the lower bound
\begin{align}
    R(\Psi_\alpha)
    \ge
    \frac{
    1-\sum_{c\in\mathcal C}|\alpha_c|^4
    }{
    \beta_{\rm coh}^{(\alpha)}
    }.
    \label{eq:weighted-algorithm-goal-rom}
\end{align}

The graph-basis reduction shows that every stabilizer state contributes a
value of the form
\begin{align}
    F(D,\chi)
    =
    2^{-r}
    \left[
    \left|
    \sum_{c\in D}\alpha_c^*\chi(c)
    \right|^2
    -
    \sum_{c\in D}|\alpha_c|^2
    \right],
    \label{eq:weighted-algorithm-objective}
\end{align}
where \(D=\mathcal C\cap A\), \(A=t+V\) is an affine support,
\(r=\dim V\), and \(\chi\) is an affine-quadratic phase on \(A\).  The
algorithm computes the maximum of \(|F(D,\chi)|\) exactly as follows.

\paragraph*{Algorithm.}
Initialize \(\beta=0\).  Then perform the following steps.

\begin{enumerate}[leftmargin=2em]
    \item Enumerate all subsets \(D\subseteq\mathcal C\).  The empty set and
    singleton sets can be skipped, since they contribute zero.

    \item For each nonempty subset
    \begin{align}
        D=\{d_1,\ldots,d_m\},
        \label{eq:weighted-algorithm-subset-D}
    \end{align}
    compute the affine hull
    \begin{align}
        A_D
        :=
        \aff(D)
        =
        d_1+
        \operatorname{span}_{\mathbb F_2}
        \{d_2+d_1,\ldots,d_m+d_1\}.
        \label{eq:weighted-algorithm-affine-hull}
    \end{align}
    Write
    \begin{align}
        A_D=t_D+V_D,
        \qquad
        r_D:=\dim V_D.
        \label{eq:weighted-algorithm-affine-dimension}
    \end{align}

    \item Check whether \(D\) is an exact affine intersection with the CWS
    word set:
    \begin{align}
        A_D\cap\mathcal C=D.
        \label{eq:weighted-algorithm-valid-intersection}
    \end{align}
    If this condition fails, discard \(D\).  Indeed, every affine support
    containing \(D\) must also contain \(\aff(D)\), and therefore must also
    contain the extra CWS words in \(A_D\cap\mathcal C\).

    \item Choose a basis
    \begin{align}
        v_1,\ldots,v_{r_D}
        \label{eq:weighted-algorithm-basis}
    \end{align}
    of \(V_D\).  For every \(c\in D\), compute its coordinate
    \(y(c)\in\mathbb F_2^{r_D}\) defined by
    \begin{align}
        c
        =
        t_D+
        \sum_{j=1}^{r_D}y_j(c)v_j.
        \label{eq:weighted-algorithm-coordinate}
    \end{align}

    \item Enumerate all affine-quadratic phases on \(A_D\), modulo an
    irrelevant global phase.  In the coordinates \(y\in\mathbb F_2^{r_D}\),
    these phases may be written as
    \begin{align}
        \chi_{a,b}(y)
        =
        i^{\sum_{j=1}^{r_D}a_jy_j}
        (-1)^{
        \sum_{1\le j<k\le r_D}b_{jk}y_jy_k
        },
        \label{eq:weighted-phase-parameterization}
    \end{align}
    where $a_j\in\mathbb Z_4,b_{jk}\in\mathbb F_2.$

    \item For each phase \(\chi_{a,b}\), evaluate
    \begin{align}
        F(D,a,b)
        =
        2^{-r_D}
        \left[
        \left|
        \sum_{c\in D}
        \alpha_c^*\chi_{a,b}(y(c))
        \right|^2
        -
        \sum_{c\in D}|\alpha_c|^2
        \right].
        \label{eq:weighted-algorithm-candidate-value}
    \end{align}
    Update
    \begin{align}
        \beta
        \leftarrow
        \max\{\beta,|F(D,a,b)|\}.
        \label{eq:weighted-algorithm-update}
    \end{align}
\end{enumerate}

After all subsets and all phases are tested, the algorithm returns
\begin{align}
    \beta_{\rm coh}^{(\alpha)}=\beta.
    \label{eq:weighted-algorithm-output}
\end{align}

We now estimate the running time.  There are \(2^K\) possible subsets
\(D\subseteq\mathcal C\).  For each subset, computing \(\aff(D)\), finding a
basis of \(V_D\), testing membership in \(A_D\cap\mathcal C\), and computing
the coordinates \(y(c)\) can all be done by Gaussian elimination over
\(\mathbb F_2\).  A crude upper bound for this linear-algebraic part is
\begin{align}
    O(n^3+Kn^2)
    \label{eq:gaussian-complexity-bound}
\end{align}
per subset.

For a valid subset \(D\), the affine dimension satisfies
\begin{align}
    r_D
    =
    \dim\aff(D)
    \le
    |D|-1
    \le
    K-1.
    \label{eq:weighted-affine-dimension-bound}
\end{align}
The number of phases in Eq.~\eqref{eq:weighted-phase-parameterization} is
\begin{align}
    4^{r_D}2^{\binom{r_D}{2}}
    \le
    4^K2^{\binom K2}.
    \label{eq:phase-count-bound}
\end{align}
For each phase, evaluating Eq.~\eqref{eq:weighted-algorithm-candidate-value}
requires a sum over at most \(K\) CWS words.  Therefore the total running time
is bounded by
\begin{align}
    T_{\rm wt}(n,K)
    &\le
    O\!\left(
    2^K(n^3+Kn^2)
    +
    2^K K\,4^K2^{\binom K2}
    \right)
    \notag\\
    &=
    O\!\left(
    2^K(n^3+Kn^2)
    +
    K2^{\binom K2+3K}
    \right).
    \label{eq:weighted-runtime-appendix}
\end{align}
This is the claimed fixed-parameter algorithm in the number \(K=|\mathcal C|\)
of CWS words.

\subsection{Complexity of the equal-weight algorithm}
\label{subsec:equal-complexity-proof}

For equal weights,
\begin{align}
    \alpha_c=\frac1{\sqrt K},
    \qquad c\in\mathcal C,
    \label{eq:equal-weight-alpha-complexity}
\end{align}
the phase enumeration in the weighted algorithm is unnecessary.  Indeed, for
a fixed affine support \(A\), let
\begin{align}
    D=A\cap\mathcal C,
    \qquad
    m=|D|,
    \qquad
    r=\dim A.
    \label{eq:equal-algorithm-D-m-r}
\end{align}
The graph-basis formula gives
\begin{align}
    \bra{\phi}B_{\rm coh}\ket{\phi}
    =
    \frac{2^{-r}}{K}
    \left[
    \left|
    \sum_{c\in D}\chi(c)
    \right|^2
    -
    |D|
    \right].
    \label{eq:equal-algorithm-phase-formula}
\end{align}
Since \(|\chi(c)|=1\), one has
\begin{align}
    \left|
    \sum_{c\in D}\chi(c)
    \right|
    \le
    m.
    \label{eq:equal-phase-upper-bound}
\end{align}
This bound is achieved by the constant phase \(\chi(c)=1\), which is a valid
affine-quadratic phase.  Therefore, for equal weights, the best value for a
fixed affine support is
\begin{align}
    \frac{m(m-1)}{K2^r}.
    \label{eq:equal-fixed-support-best-value}
\end{align}

Thus the equal-weight algorithm is simpler.

\paragraph*{Equal-weight algorithm.}
Initialize \(\beta=0\).  Enumerate all subsets \(D\subseteq\mathcal C\).  For
each subset, compute
\begin{align}
    A_D=\aff(D),
    \qquad
    r_D=\dim A_D.
    \label{eq:equal-algorithm-affine-hull}
\end{align}
Check whether
\begin{align}
    A_D\cap\mathcal C=D.
    \label{eq:equal-algorithm-valid-D}
\end{align}
If not, discard \(D\).  If yes, evaluate
\begin{align}
    G(D)
    :=
    \frac{|D|(|D|-1)}{K2^{r_D}}.
    \label{eq:equal-subset-candidate-value}
\end{align}
Update
\begin{align}
    \beta\leftarrow \max\{\beta,G(D)\}.
    \label{eq:equal-algorithm-update}
\end{align}
At the end, the output is
\begin{align}
    \beta_{\rm coh}=\beta.
    \label{eq:equal-algorithm-output}
\end{align}

The same Gaussian-elimination step as above gives a per-subset cost bounded
by Eq.~\eqref{eq:gaussian-complexity-bound}.  Since there are \(2^K\) subsets
and no phase enumeration, the total running time is
\begin{align}
    T_{\rm eq}(n,K)
    =
    O\!\left(2^K(n^3+Kn^2)\right).
    \label{eq:equal-subset-runtime-appendix}
\end{align}

Equivalently, the equal-weight threshold can be computed from the
affine-intersection profile
\begin{align}
    M_r(\mathcal C)
    =
    \max_{\substack{A\subseteq\mathbb F_2^n\\ A\ {\rm affine}\\ \dim A=r}}
    |A\cap\mathcal C|,
    \label{eq:equal-complexity-affine-profile}
\end{align}
through
\begin{align}
    \beta_{\rm coh}
    =
    \max_{0\le r\le n}
    \frac{
    M_r(\mathcal C)\bigl(M_r(\mathcal C)-1\bigr)
    }{
    K2^r
    }.
    \label{eq:equal-complexity-profile-formula}
\end{align}
One may alternatively enumerate affine flats directly.  Since the number of
\(r\)-dimensional affine flats in \(\mathbb F_2^n\) is
\begin{align}
    2^{n-r}{n\brack r}_2,
    \label{eq:number-of-affine-flats}
\end{align}
where \({n\brack r}_2=\prod_{i=0}^{r-1}\frac{2^n-2^i}{2^r-2^i}\) is the binary Gaussian binomial coefficient
\footnote{The coefficient \({n\brack r}_2\) counts the number of \(r\)-dimensional linear subspaces of \(\mathbb F_2^n\). Indeed, the numerator counts ordered choices of \(r\) linearly independent vectors in \(\mathbb F_2^n\), while the denominator accounts for the different ordered bases spanning the same \(r\)-dimensional subspace. Each such linear subspace has \(2^{n-r}\) distinct affine translates, corresponding to the cosets of the subspace in \(\mathbb F_2^n\), which gives the factor \(2^{n-r}\) above.}, the
direct affine-flat enumeration has running time
\begin{align}
    T_{\rm flat}(n,K)
    =
    O\!\left(
    K\sum_{r=0}^n 2^{n-r}{n\brack r}_2
    \right)
    =
    2^{n^2/4+O(n)}K.
    \label{eq:affine-flat-runtime-appendix}
\end{align}
For small \(K\), subset enumeration is preferable; for small \(n\) but larger
\(K\), affine-flat enumeration may be faster.

\section{Calculation details, Bell checks, and Pauli expansions}
\label{app:calculation-details}

\subsection{Bell-type numerical checks and threshold reductions}
\label{subsec:bell-checks}

This appendix records the Bell-inspired numerical checks shown in
Fig.~\ref{fig:combined-bell-examples}.  These checks are not used in the CWS
affine-profile theorem.  Their role is to illustrate the general principle
of Theorem~\ref{thm:witness-rom-bound}: once a calibrated Hermitian
observable has a known stabilizer threshold, its expectation value gives a
certified lower bound on RoM.

The numerical implementation performs three related checks.  The first one uses the tilted-CHSH family introduced in Ref.~\cite{AcinMassarPironio2012Tilted} and further analyzed in the self-testing literature~\cite{YangNavascues2013Tilted,BampsPironio2015Tilted}.  First, for the two-qubit
family
\begin{align}
    \ket{\psi_\theta}
    =
    \cos\theta\ket{00}+\sin\theta\ket{11},
    \label{eq:two-qubit-family-appendix}
\end{align}
it compares the exact RoM with the tilted-CHSH lower bound.  The tilt
parameter is chosen as
\begin{align}
    \alpha(\theta)
    =
    \frac{2}{\sqrt{1+2\tan^2(2\theta)}},
    \label{eq:tilted-alpha-theta-appendix}
\end{align}
equivalently implemented in the numerically stable form
\begin{align}
    \alpha(\theta)
    =
    \frac{2\cos(2\theta)}
    {\sqrt{\cos^2(2\theta)+2\sin^2(2\theta)}}.
    \label{eq:tilted-alpha-stable-appendix}
\end{align}
For this convention, the tilted-CHSH quantum value and stabilizer threshold
used in the plot are
\begin{align}
    Q_{\rm tilt}(\alpha)
    &=
    \sqrt{8+2\alpha^2},
    \label{eq:tilted-quantum-value-appendix}\\
    \beta_{\rm stab}^{\rm tilt}(\alpha)
    &=
    \max(2+\alpha,2\sqrt2).
    \label{eq:tilted-stabilizer-threshold-appendix}
\end{align}
Thus the plotted lower bound is
\begin{align}
    R_{\rm lb}(\theta)
    =
    \frac{Q_{\rm tilt}(\alpha(\theta))}
    {\beta_{\rm stab}^{\rm tilt}(\alpha(\theta))}
    =
    \frac{\sqrt{8+2\alpha(\theta)^2}}
    {\max(2+\alpha(\theta),2\sqrt2)}.
    \label{eq:tilted-chsh-rom-appendix}
\end{align}

Second, following the three-qubit Bell inequality involving two- and three-body correlations introduced for GHZ and W states~\cite{Cabello2002GHZW} and later used as a magic-witness example~\cite{Macedo2025}, the numerical implementation evaluates the standard three-qubit W state and a one-parameter family of generalized W states,
\begin{align}
    \ket{W_G(\theta,\phi)}
    =
    \sin\theta\sin\phi\ket{001}
    +
    \sin\theta\cos\phi\ket{010}
    +
    \cos\theta\ket{100}.
    \label{eq:w-state-appendix}
\end{align}
The line shown in Fig.~\ref{fig:3-qubit} fixes
\begin{align}
    \phi=\frac{\pi}{4},
    \label{eq:w-line-phi-appendix}
\end{align}
and scans \(\theta\).  The Bell-type functional is
\begin{align}
    B
    =
    S_3+R_2,
    \qquad
    \beta_{\rm stab}(B)=6.
    \label{eq:w-state-bell-operator-appendix}
\end{align}
Here \(S_3\) is the signed three-body correlator part and \(R_2\) is the
two-body correlator part.  In terms of two local dichotomic measurements
\(A_i^{(0)},A_i^{(1)}\) for each qubit, we use the following convention for the signed three-body part, consistent with the Bell functional discussed in Refs.~\cite{Cabello2002GHZW,Macedo2025}:
\begin{align}
S_3
&=
\sum_{\boldsymbol{x}\in\{0,1\}^3}
(-1)^{x_1x_2+x_1x_3+x_2x_3+x_1+x_2+x_3}
\left\langle
\prod_{i=1}^3 A_i^{(x_i)}
\right\rangle .
\end{align}
where
\begin{align}
    R_2
    &=
    \langle A_1^{(0)}A_2^{(1)}\rangle
    +
    \langle A_1^{(0)}A_3^{(1)}\rangle
    +
    \langle A_2^{(0)}A_3^{(1)}\rangle
    \notag\\
    &\quad
    +
    \langle A_1^{(1)}A_2^{(0)}\rangle
    +
    \langle A_1^{(1)}A_3^{(0)}\rangle
    +
    \langle A_2^{(1)}A_3^{(0)}\rangle.
    \label{eq:r2-definition-appendix}
\end{align}
For each state, the local measurement directions are optimized numerically
over the Bloch sphere, giving an optimized value \(B^\star\).  The plotted
three-qubit lower bound is therefore
\begin{align}
    R_{\rm lb}
    =
    \frac{B^\star}{6}.
    \label{eq:w-rom-lower-bound-appendix}
\end{align}

Finally, the numerical implementation computes exact RoM for \(n\le3\) by solving the stabilizer
decomposition linear program.  It generates all pure stabilizer states by a
Clifford-orbit breadth-first search and solves
\begin{align}
    \min_x \|x\|_1
    \quad
    \text{subject to}
    \quad
    \rho=\sum_j x_j\sigma_j,
    \label{eq:exact-rom-lp-appendix}
\end{align}
where \(\sigma_j\) range over pure stabilizer states.  The expected counts
are \(6\), \(60\), and \(1080\) pure stabilizer states for one, two, and
three qubits, respectively.  The blue curves in Fig.~\ref{fig:combined-bell-examples}
come from this exact LP, while the orange curves come from the corresponding
witness lower bounds.

\subsection{CWS data used in the affine-profile computations}
\label{subsec:cws-data-computations}

For reproducibility, we record the CWS data used in Table~\ref{tab:cws-high-distance-examples}.  The equal-weight threshold only depends on the classical word set \(\mathcal C\).  The graph stabilizer specifies the physical CWS code and the Pauli realization of the coherence witness.

\paragraph{The \texorpdfstring{\(((5,6,2))\)}{((5,6,2))} five-cycle code.}
The graph is the five-cycle with stabilizer generators
\begin{align}
    g_i
    =
    X_iZ_{i-1}Z_{i+1},
    \qquad
    i\in\mathbb Z_5,
    \label{eq:five-cycle-stabilizer-appendix}
\end{align}
where indices are taken modulo five.  The word set is
\begin{align}
    \mathcal C_{5,6,2}
    =
    \{00000,11010,01101,10110,01011,10101\}.
    \label{eq:562-wordset-appendix}
\end{align}
The numerator is
\begin{align}
    \bra{\Psi_{\mathcal C}}B_{\rm coh}\ket{\Psi_{\mathcal C}}
    =
    1-\frac16
    =
    \frac56,
    \label{eq:562-numerator-appendix}
\end{align}
and the affine-profile computation gives \((r_*,m_*)=(2,3)\), hence
\begin{align}
    \beta_{\rm coh}
    =
    \frac{3\cdot2}{6\cdot2^2}
    =
    \frac14.
    \label{eq:562-beta-appendix}
\end{align}
Therefore \(R\ge10/3\).

\paragraph*{The \texorpdfstring{\(((9,12,3))\)}{((9,12,3))} cyclic CWS code.}
The graph is the nine-cycle, with generators \(g_i=X_iZ_{i-1}Z_{i+1}\).  We use the CWS word set $\mathcal C_{9,12,3}$
\begin{align}
=\bigl\{&
000000000,100100100,010001100,110101000,
\notag\\
&000110001,100010101,011001010,111101110,
\notag\\
&001010011,101110111,011111111,111011011
\bigr\}.
\label{eq:912-wordset-appendix}
\end{align}

\paragraph*{The \texorpdfstring{\(((10,18,3))\)}{((10,18,3))} ring CWS code.}
The graph is the ten-cycle, again with generators \(g_i=X_iZ_{i-1}Z_{i+1}\).  The word set is $\mathcal C_{10,18,3}$
\begin{align}=\bigl\{&
0000000000,1101001100,0011001010,0000011111,
\notag\\
&0010001001,1111100000,1000111110,1100100101,
\notag\\
&0101101101,0001000110,1010010010,0100110100,
\notag\\
&1001010111,1011010001,0110111000,0101110010,
\notag\\
&1110100011,0111111011
\bigr\}.
\label{eq:1018-wordset-appendix}
\end{align}

\paragraph*{The \texorpdfstring{\(((10,20,3))\)}{((10,20,3))} double-ring CWS code.}
For the double-ring code we use the stabilizer generators listed in Ref.~\cite{Cross2009}:
\begin{align}
    S_1&=\mathsf{XZIIZZIIII}, & S_6&=\mathsf{ZIIIIXZIIZ},
    \notag\\
    S_2&=\mathsf{ZXZIIIZIII}, & S_7&=\mathsf{IZIIIZXZII},
    \notag\\
    S_3&=\mathsf{IZXZIIIZII}, & S_8&=\mathsf{IIZIIIZXZI},
    \notag\\
    S_4&=\mathsf{IIZXZIIIZI}, & S_9&=\mathsf{IIIZIIIZXZ},
    \notag\\
    S_5&=\mathsf{ZIIZXIIIIZ}, & S_{10}&=\mathsf{IIIIZZIIZX}.
    \label{eq:double-ring-stabilizers-appendix}
\end{align}
The word set is $\mathcal C_{10,20,3}$
\begin{align}
=\bigl\{&
0000000000,1100101101,1100000100,0010010010,
\notag\\
&1001100100,0111011011,1101111110,0010111011,
\notag\\
&1001101111,0111010000,1111000101,1011010100,
\notag\\
&0101100000,1011011111,0101101011,0011000001,
\notag\\
&0000101001,1110010110,0001111010,1110111111
\bigr\}.
\label{eq:1020-wordset-appendix}
\end{align}

\paragraph*{The SSW code.}
The SSW word set used in the last row of Table~\ref{tab:cws-high-distance-examples} is generated by Eq.~\eqref{eq:ssw-standard-label-main}.  For \(n=7\), the resulting affine profile is Eq.~\eqref{eq:ssw-seven-profile-main}.

\section{NISQ implementation, measurement cost, and noise tolerance}
\label{app:nisq-discussion}

The affine-profile formula gives the stabilizer threshold of the CWS
coherence witness, but an experimental or NISQ use of the result requires
one more step: the witness must be expressed as a measurable Pauli
observable.  This is where the CWS structure is useful again.  The witness is
not treated as a black-box projector.  Instead, its Pauli expansion follows
directly from the graph-state stabilizer expansion
\begin{align}
    \ket{G}\bra{G}
    =
    2^{-n}
    \sum_{s\in S_G}s,
    \label{eq:nisq-graph-projector}
\end{align}
and hence
\begin{align}
    \ket{c_G}\bra{c'_G}
    =
    2^{-n}
    \sum_{s\in S_G}Z^c s Z^{c'}.
    \label{eq:nisq-transition-pauli}
\end{align}
Substituting Eq.~\eqref{eq:nisq-transition-pauli} into the coherence witness
gives a finite Pauli expansion whose coefficients are determined entirely by
the graph \(G\), the word set \(\mathcal C\), and the amplitudes
\(\alpha\).  Therefore, once the affine-profile threshold has been computed,
the certification protocol is conceptually simple: prepare the coherent CWS
superposition, estimate the Pauli correlators appearing in
\(B_{\rm coh}^{(\alpha)}\), and compare the observed value with the
stabilizer threshold
\begin{align}
    \beta_{\rm coh}^{(\alpha)}
    =
    \max_{\sigma\in\STAB_n}
    \left|
    \Tr\!\left(B_{\rm coh}^{(\alpha)}\sigma\right)
    \right|.
    \label{eq:nisq-stabilizer-threshold}
\end{align}
If the observed value exceeds this threshold, Theorem~\ref{thm:witness-rom-bound}
gives a quantitative lower bound on robustness of magic.

The \(((5,6,2))\) code provides a concrete small-scale illustration.  In
that case the collected Pauli expansion of \(B_{\rm coh}\) contains \(240\)
nonzero Pauli strings with coefficients \(\pm 1/96\).  This number should be
interpreted only as a naive correlator count, not as an optimized
measurement-setting count.  Many Pauli strings can be grouped into common
local Pauli bases, and classical-shadow or derandomized-shadow methods may
further reduce the number of experimental settings.  The important point is
that the witness is a Pauli observable with an analytically known stabilizer
threshold; it does not require full state tomography.

\subsection{A global depolarizing benchmark}
\label{subsec:general-noise-tolerance}

We next record a simple noise benchmark for the same measurement protocol.
The purpose of this calculation is not to model a realistic device, but to
show how the witness value degrades under a standard global depolarizing
channel.  The calculation is general and applies to any weighted CWS
superposition.

First note that the coherence witness is traceless.  Indeed,
\begin{align}
    \Tr B_{\rm coh}^{(\alpha)}
    &=
    \Tr\!\left(\ket{\Psi_\alpha}\bra{\Psi_\alpha}\right)
    -
    \Tr\!\left(
    \sum_{c\in\mathcal C}
    |\alpha_c|^2
    \ket{c_G}\bra{c_G}
    \right)
    \notag\\
    &=
    1-
    \sum_{c\in\mathcal C}|\alpha_c|^2
    =
    0.
    \label{eq:coherence-witness-traceless}
\end{align}
Thus the maximally mixed component contributes no signal.  Let
\begin{align}
    \rho_p
    =
    (1-p)\ket{\Psi_\alpha}\bra{\Psi_\alpha}
    +
    p\frac{I}{2^n}
    \label{eq:global-depolarized-state}
\end{align}
be the globally depolarized target state, and define the ideal numerator
\begin{align}
    N_\alpha
    :=
    \bra{\Psi_\alpha}
    B_{\rm coh}^{(\alpha)}
    \ket{\Psi_\alpha}
    =
    1-
    \sum_{c\in\mathcal C}|\alpha_c|^4 .
    \label{eq:noise-numerator-alpha}
\end{align}
Then the measured witness value under global depolarizing noise is simply
rescaled:
\begin{align}
    \Tr\!\left(B_{\rm coh}^{(\alpha)}\rho_p\right)
    &=
    (1-p)
    \bra{\Psi_\alpha}
    B_{\rm coh}^{(\alpha)}
    \ket{\Psi_\alpha}
    +
    p\,2^{-n}\Tr B_{\rm coh}^{(\alpha)}
    \notag\\
    &=
    (1-p)N_\alpha .
    \label{eq:depolarized-witness-value}
\end{align}
Combining this with the witness lower bound gives
\begin{align}
    R(\rho_p)
    \ge
    \frac{(1-p)N_\alpha}{\beta_{\rm coh}^{(\alpha)}}.
    \label{eq:depolarized-rom-bound}
\end{align}
The same witness remains nontrivial as long as the noisy expectation still
exceeds the stabilizer threshold, namely
\begin{align}
    (1-p)N_\alpha
    >
    \beta_{\rm coh}^{(\alpha)}.
    \label{eq:depolarizing-nontrivial-condition}
\end{align}
Equivalently,
\begin{align}
    p
    <
    1-
    \frac{\beta_{\rm coh}^{(\alpha)}}{N_\alpha}.
    \label{eq:depolarizing-threshold-general}
\end{align}

For equal weights, \(N_\alpha=1-1/K\).  Thus every row of
Table~\ref{tab:cws-high-distance-examples} immediately gives a corresponding
global-depolarizing tolerance.  For example, for the \(((5,6,2))\) CWS state,
one has
\begin{align}
    K=6,
    \qquad
    N_\alpha=\frac56,
    \qquad
    \beta_{\rm coh}=\frac14.
    \label{eq:562-noise-inputs}
\end{align}
Therefore the witness remains nontrivial whenever
\begin{align}
    p
    <
    1-
    \frac{1/4}{5/6}
    =
    \frac{7}{10},
    \label{eq:562-depolarizing-threshold}
\end{align}
and the noisy robustness lower bound is
\begin{align}
    R(\rho_p)
    \ge
    \frac{10}{3}(1-p).
    \label{eq:562-noisy-rom-bound}
\end{align}
This example should be viewed only as a benchmark calculation.  Realistic
local noise, coherent control errors, and readout errors require a separate
analysis.

\subsection{Relation to measurement-limited magic certification}
\label{subsec:nisq-reduced-polytope-comparison}

The present approach is complementary to recent measurement-limited
approaches to magic certification.  In Ref.~\cite{liu2026graph},
one begins with a chosen set of Pauli measurements and studies the projection
of the stabilizer polytope onto the corresponding measured coordinates.  The
tractability of that reduced problem is controlled by the frustration graph
of the measurement set together with Pauli sign-dependency constraints.

Here the order of construction is different.  We do not begin from an
arbitrary Pauli measurement set.  Instead, the CWS code first selects a
structured coherence observable \(B_{\rm coh}^{(\alpha)}\).  The Pauli
measurement set is then inherited from the expansion of this observable, and
the stabilizer threshold is controlled by the affine-incidence geometry of
the CWS word set.  Thus both approaches replace full stabilizer-polytope
optimization by a combinatorial problem, but the relevant combinatorial
objects are different: frustration graphs for measurement-limited reduced
RoM, and affine flats for CWS codeword coherence.

There are also important limitations.  The certificate discussed here
assumes calibrated Pauli measurements and is therefore a magic witness, not a
device-independent Bell inequality \cite{acin2006bell, acin2007device}.  A genuine Bell version would require
comparison with a local-hidden-variable threshold.  Moreover, the witness is
tailored to a chosen CWS word set and is most meaningful when the target
preparation is the coherent CWS superposition, rather than an arbitrary state
in the code space.

A natural future direction is to search for Bell inequalities while
preserving the coherence-witness structure.  Instead of optimizing over
arbitrary Bell operators, one may restrict to Pauli observables satisfying
\begin{align}
    \bra{c_G}B\ket{c_G}=0,
    \qquad
    c\in\mathcal C.
    \label{eq:coherence-preserving-bell-constraint}
\end{align}
This constraint ensures that the Bell functional is sensitive to codeword
coherence rather than diagonal codeword population.  Such a constrained
search would preserve the CWS mechanism developed in this work, unlike a
generic Bell-operator optimization over the same Pauli support.

\bibliographystyle{unsrt}
\bibliography{sample}

@article{gross2006hudson,
  title={Hudson’s theorem for finite-dimensional quantum systems},
  author={Gross, David},
  journal={Journal of mathematical physics},
  volume={47},
  number={12},
  year={2006},
  publisher={AIP Publishing}
}

@inproceedings{mckague2011self,
  title={Self-testing graph states},
  author={McKague, Matthew},
  booktitle={Conference on Quantum Computation, Communication, and Cryptography},
  pages={104--120},
  year={2011},
  organization={Springer}
}

@article{brunner2012testing,
  title={Testing the structure of multipartite entanglement with Bell inequalities},
  author={Brunner, Nicolas and Sharam, James and V{\'e}rtesi, Tam{\'a}s},
  journal={Physical review letters},
  volume={108},
  number={11},
  pages={110501},
  year={2012},
  publisher={APS}
}

@article{kuo2024self,
  title={Self-Testing Quantum Error Correcting Codes: Analyzing Computational Hardness},
  author={Kuo, En-Jui and Hsu, Li-Yi},
  journal={arXiv preprint arXiv:2409.01987},
  year={2024}
}

@article{bowles2018device,
  title={Device-independent entanglement certification of all entangled states},
  author={Bowles, Joseph and {\v{S}}upi{\'c}, Ivan and Cavalcanti, Daniel and Ac{\'\i}n, Antonio},
  journal={Physical review letters},
  volume={121},
  number={18},
  pages={180503},
  year={2018},
  publisher={APS}
}

@article{hostens2005stabilizer,
  title={Stabilizer states and Clifford operations for systems of arbitrary dimensions and modular arithmetic},
  author={Hostens, Erik and Dehaene, Jeroen and De Moor, Bart},
  journal={Physical Review A—Atomic, Molecular, and Optical Physics},
  volume={71},
  number={4},
  pages={042315},
  year={2005},
  publisher={APS}
}

@article{rains1999quantum,
  title={Quantum codes of minimum distance two},
  author={Rains, Eric M.},
  journal={IEEE Transactions on Information theory},
  volume={45},
  number={1},
  pages={266--271},
  year={1999},
  publisher={IEEE}
}

@article{rains1997nonadditive,
  title={A nonadditive quantum code},
  author={Rains, Eric M and Hardin, RH and Shor, Peter W and Sloane, Neil JA},
  journal={arXiv preprint quant-ph/9703002},
  year={1997}
}

@incollection{eczoo_rains,
 title={\(((2m+1,3 \times 2^{2m-3},2))\) Rains code},
 booktitle={The Error Correction Zoo},
 year={2026},
 editor={Albert, Victor V. and Faist, Philippe},
 eprint={2606.11484},
 doi={10.48550/arXiv.2606.11484},
 url={https://errorcorrectionzoo.org/c/rains}
}

@article{crew2025magic,
  title={Magic entropy in hybrid spin-boson systems},
  author={Crew, Samuel and Li, Ying-Lin and Li, Heng-Hsi and Chang, Po-Yao},
  journal={Reports on Progress in Physics},
  year={2025}
}

@article{hamaguchi2024handbook,
  title={Handbook for quantifying robustness of magic},
  author={Hamaguchi, Hiroki and Hamada, Kou and Yoshioka, Nobuyuki},
  journal={Quantum},
  volume={8},
  pages={1461},
  year={2024},
  publisher={Verein zur F{\"o}rderung des Open Access Publizierens in den Quantenwissenschaften}
}

@article{hughes2014quantum,
  title={Quantum non-Gaussianity witnesses in phase space},
  author={Hughes, Catherine and Genoni, Marco G and Tufarelli, Tommaso and Paris, Matteo GA and Kim, MS},
  journal={Physical Review A},
  volume={90},
  number={1},
  pages={013810},
  year={2014},
  publisher={APS}
}

@article{takahashi1986wigner,
  title={Wigner and Husimi functions in quantum mechanics},
  author={Takahashi, Kin'ya},
  journal={Journal of the Physical Society of Japan},
  volume={55},
  number={3},
  pages={762--779},
  year={1986},
  publisher={The Physical Society of Japan}
}

@article{harriman1988some,
  title={Some properties of the Husimi function},
  author={Harriman, John E},
  journal={The Journal of chemical physics},
  volume={88},
  number={10},
  pages={6399--6408},
  year={1988},
  publisher={American Institute of Physics}
}

@incollection{eczoo_qubit_9_12_3,
 title={\(((9,12,3))\) qubit code},
 booktitle={The Error Correction Zoo},
 year={2026},
 editor={Albert, Victor V. and Faist, Philippe},
 eprint={2606.11484},
 doi={10.48550/arXiv.2606.11484},
 url={https://errorcorrectionzoo.org/c/qubit_9_12_3}
}

@incollection{eczoo_arvind,
 title={\(((n,1+n(q-1),2))_q\) union stabilizer code},
 booktitle={The Error Correction Zoo},
 year={2026},
 editor={Albert, Victor V. and Faist, Philippe},
 eprint={2606.11484},
 doi={10.48550/arXiv.2606.11484},
 url={https://errorcorrectionzoo.org/c/arvind}
}

@article{weinbub2018recent,
  title={Recent advances in Wigner function approaches},
  author={Weinbub, Josef and Ferry, DK},
  journal={Applied Physics Reviews},
  volume={5},
  number={4},
  year={2018},
  publisher={AIP Publishing}
}

@article{mele2026symplectic,
  title={Symplectic Rank of Non-Gaussian Quantum States},
  author={Mele, Francesco A and Oliviero, Salvatore FE and Upreti, Varun and Chabaud, Ulysse},
  journal={PRX Quantum},
  volume={7},
  number={2},
  pages={020366},
  year={2026},
  publisher={APS}
}

@article{hahn2026assessing,
  title={Assessing non-Gaussian quantum state conversion with the stellar rank},
  author={Hahn, Oliver and Garnier, Maxime and Ferrini, Giulia and Ferraro, Alessandro and Chabaud, Ulysse},
  journal={Quantum},
  volume={10},
  pages={2095},
  year={2026},
  publisher={Verein zur F{\"o}rderung des Open Access Publizierens in den Quantenwissenschaften}
}

@article{chabaud2020stellar,
  title={Stellar representation of non-Gaussian quantum states},
  author={Chabaud, Ulysse and Markham, Damian and Grosshans, Fr{\'e}d{\'e}ric},
  journal={Physical Review Letters},
  volume={124},
  number={6},
  pages={063605},
  year={2020},
  publisher={APS}
}

@article{gottesman1998heisenberg,
  title={The Heisenberg representation of quantum computers},
  author={Gottesman, Daniel},
  journal={arXiv preprint quant-ph/9807006},
  year={1998}
}

@article{aaronson2004improved,
  title={Improved simulation of stabilizer circuits},
  author={Aaronson, Scott and Gottesman, Daniel},
  journal={Physical Review A—Atomic, Molecular, and Optical Physics},
  volume={70},
  number={5},
  pages={052328},
  year={2004},
  publisher={APS}
}

@article{webster2022universal,
  title={Universal fault-tolerant quantum computing with stabilizer codes},
  author={Webster, Paul and Vasmer, Michael and Scruby, Thomas R and Bartlett, Stephen D},
  journal={Physical Review Research},
  volume={4},
  number={1},
  pages={013092},
  year={2022},
  publisher={APS}
}

@misc{LiuXu2026CWSMagic,
  title        = {Quantifying Nonstabilizerness of Codeword-Stabilized Codes},
  author       = {Liu, Yuan and Xu, Ke-Mi},
  year         = {2026},
  eprint       = {2608.22017},
  archivePrefix= {arXiv},
  primaryClass = {quant-ph}
}

@article{AcinMassarPironio2012Tilted,
  title        = {Randomness versus Nonlocality and Entanglement},
  author       = {Ac{\'i}n, Antonio and Massar, Serge and Pironio, Stefano},
  journal      = {Physical Review Letters},
  volume       = {108},
  pages        = {100402},
  year         = {2012},
  doi          = {10.1103/PhysRevLett.108.100402}
}

@article{BampsPironio2015Tilted,
  title        = {Sum-of-Squares Decompositions for a Family of {Clauser-Horne-Shimony-Holt}-Like Inequalities and Their Application to Self-Testing},
  author       = {Bamps, C{\'e}dric and Pironio, Stefano},
  journal      = {Physical Review A},
  volume       = {91},
  pages        = {052111},
  year         = {2015},
  doi          = {10.1103/PhysRevA.91.052111}
}

@article{YangNavascues2013Tilted,
  title        = {Robust Self-Testing of Unknown Quantum Systems into Any Entangled Two-Qubit States},
  author       = {Yang, Tzyh Haur and Navascu{\'e}s, Miguel},
  journal      = {Physical Review A},
  volume       = {87},
  pages        = {050102},
  year         = {2013},
  doi          = {10.1103/PhysRevA.87.050102}
}

@article{liu2026graph,
  title={Graph Theoretic Approach to Quantum Nonstabilizerness},
  author={Liu, Yingjian and Gasull, Albert and Hu, Mengyao and Zhang, Ruiyun and Baccari, Flavio and Tura, Jordi},
  journal={arXiv preprint arXiv:2607.26154},
  year={2026}
}

@article{Cabello2002GHZW,
  title        = {Bell's Theorem with and without Inequalities for the Three-Qubit {Greenberger-Horne-Zeilinger} and {W} States},
  author       = {Cabello, Ad{\'a}n},
  journal      = {Physical Review A},
  volume       = {65},
  pages        = {032108},
  year         = {2002},
  doi          = {10.1103/PhysRevA.65.032108}
}

@article{acin2007device,
  title={Device-independent security of quantum cryptography against collective attacks},
  author={Ac{\'\i}n, Antonio and Brunner, Nicolas and Gisin, Nicolas and Massar, Serge and Pironio, Stefano and Scarani, Valerio},
  journal={Physical Review Letters},
  volume={98},
  number={23},
  pages={230501},
  year={2007},
  doi={10.1103/PhysRevLett.98.230501}
}

@article{acin2006bell,
  title={From Bell’s theorem to secure quantum key distribution},
  author={Acin, Antonio and Gisin, Nicolas and Masanes, Lluis},
  journal={Physical review letters},
  volume={97},
  number={12},
  pages={120405},
  year={2006},
  doi={10.1103/PhysRevLett.97.120405}
}

@incollection{eczoo_qubit_5_6_2,
 title={\(((5,6,2))\) qubit code},
 booktitle={The Error Correction Zoo},
 year={2026},
 editor={Albert, Victor V. and Faist, Philippe},
 eprint={2606.11484},
 doi={10.48550/arXiv.2606.11484},
 url={https://errorcorrectionzoo.org/c/qubit_5_6_2}
}

@article{Horodecki2009,
  title = {Quantum entanglement},
  author = {Horodecki, Ryszard and Horodecki, Pawe\l{} and Horodecki, Micha\l{} and Horodecki, Karol},
  journal = {Rev. Mod. Phys.},
  volume = {81},
  issue = {2},
  pages = {865--942},
  numpages = {0},
  year = {2009},
  month = {Jun},
  publisher = {American Physical Society},
  doi = {10.1103/RevModPhys.81.865},
  url = {https://link.aps.org/doi/10.1103/RevModPhys.81.865}
}

@article{Streltsov2017,
  title = {Colloquium: Quantum coherence as a resource},
  author = {Streltsov, Alexander and Adesso, Gerardo and Plenio, Martin B.},
  journal = {Rev. Mod. Phys.},
  volume = {89},
  issue = {4},
  pages = {041003},
  numpages = {34},
  year = {2017},
  month = {Oct},
  publisher = {American Physical Society},
  doi = {10.1103/RevModPhys.89.041003},
  url = {https://link.aps.org/doi/10.1103/RevModPhys.89.041003}
}

@article{Brunner2014,
  title = {Bell nonlocality},
  author = {Brunner, Nicolas and Cavalcanti, Daniel and Pironio, Stefano and Scarani, Valerio and Wehner, Stephanie},
  journal = {Rev. Mod. Phys.},
  volume = {86},
  issue = {2},
  pages = {419--478},
  numpages = {60},
  year = {2014},
  month = {Apr},
  publisher = {American Physical Society},
  doi = {10.1103/RevModPhys.86.419},
  url = {https://link.aps.org/doi/10.1103/RevModPhys.86.419}
}

@article{gross2007lu,
  title={The LU-LC conjecture, diagonal local operations and quadratic forms over GF (2)},
  author={Gross, David and Nest, Maarten},
  journal={arXiv preprint arXiv:0707.4000},
  year={2007}
}

@article{dehaene2003clifford,
  title={Clifford group, stabilizer states, and linear and quadratic operations over GF (2)},
  author={Dehaene, Jeroen and De Moor, Bart},
  journal={Physical Review A},
  volume={68},
  number={4},
  pages={042318},
  year={2003},
  publisher={APS}
}

@article{PhysRevLett.118.090501,
  title = {Application of a Resource Theory for Magic States to Fault-Tolerant Quantum Computing},
  author = {Howard, Mark and Campbell, Earl},
  journal = {Phys. Rev. Lett.},
  volume = {118},
  issue = {9},
  pages = {090501},
  numpages = {6},
  year = {2017},
  month = {Mar},
  publisher = {American Physical Society},
  doi = {10.1103/PhysRevLett.118.090501},
  url = {https://link.aps.org/doi/10.1103/PhysRevLett.118.090501}
}

@article{VairogsYan2025,
  title = {Extracting randomness from magic quantum states},
  author = {Vairogs, Christopher and Yan, Bin},
  journal = {Phys. Rev. Res.},
  volume = {7},
  issue = {2},
  pages = {L022069},
  numpages = {6},
  year = {2025},
  month = {Jun},
  publisher = {American Physical Society},
  doi = {10.1103/3ttm-vhdt},
  url = {https://link.aps.org/doi/10.1103/3ttm-vhdt}
}

@article{ChowdhuryEtAl2025,
  title = {Gottesman-Knill limit on one-way communication complexity: Tracing the quantum advantage down to magic resources},
  author = {Chowdhury, Snehasish Roy and Naik, Sahil Gopalkrishna and Chakraborty, Ananya and Patra, Ram Krishna and Ghosh, Subhendu B. and Ghosal, Pratik and Banik, Manik and Maity, Ananda G.},
  journal = {Phys. Rev. A},
  volume = {113},
  issue = {4},
  pages = {042607},
  numpages = {11},
  year = {2026},
  month = {Apr},
  publisher = {American Physical Society},
  doi = {10.1103/x741-4kd1},
  url = {https://link.aps.org/doi/10.1103/x741-4kd1}
}

@article{LiuWinter2022,
  title = {Many-Body Quantum Magic},
  author = {Liu, Zi-Wen and Winter, Andreas},
  journal = {PRX Quantum},
  volume = {3},
  issue = {2},
  pages = {020333},
  numpages = {18},
  year = {2022},
  month = {May},
  publisher = {American Physical Society},
  doi = {10.1103/PRXQuantum.3.020333},
  url = {https://link.aps.org/doi/10.1103/PRXQuantum.3.020333}
}

@article{Russomanno2025,
  title = {Nonstabilizerness in the unitary and monitored quantum dynamics of XXZ-staggered and Sachdev-Ye-Kitaev models},
  author = {Russomanno, Angelo and Passarelli, Gianluca and Rossini, Davide and Lucignano, Procolo},
  journal = {Phys. Rev. B},
  volume = {112},
  issue = {6},
  pages = {064312},
  numpages = {11},
  year = {2025},
  month = {Aug},
  publisher = {American Physical Society},
  doi = {10.1103/njgn-fksh},
  url = {https://link.aps.org/doi/10.1103/njgn-fksh}
}

@article{TurkeshiDymarskySierant2025,
  title = {Pauli spectrum and nonstabilizerness of typical quantum many-body states},
  author = {Turkeshi, Xhek and Dymarsky, Anatoly and Sierant, Piotr},
  journal = {Phys. Rev. B},
  volume = {111},
  issue = {5},
  pages = {054301},
  numpages = {12},
  year = {2025},
  month = {Feb},
  publisher = {American Physical Society},
  doi = {10.1103/PhysRevB.111.054301},
  url = {https://link.aps.org/doi/10.1103/PhysRevB.111.054301}
}

@article{BeraSchiro2025,
title = {Non-stabilizerness of Sachdev-Ye-Kitaev model},
	pages = {159},
	author = {Bera, Surajit and Schirò, Marco},
	journal = {SciPost Phys.},
	volume = {19},
	year = {2025},
	publisher = {SciPost},
	doi = {10.21468/SciPostPhys.19.6.159},
	url = {https://scipost.org/10.21468/SciPostPhys.19.6.159}
}

@article{Hartse2025,
   title = {Stabilizer Scars},
  author = {Hartse, Jeremy and Fidkowski, Lukasz and Mueller, Niklas},
  journal = {Phys. Rev. Lett.},
  volume = {135},
  issue = {6},
  pages = {060402},
  numpages = {9},
  year = {2025},
  month = {Aug},
  publisher = {American Physical Society},
  doi = {10.1103/n5hb-l5p5},
  url = {https://link.aps.org/doi/10.1103/n5hb-l5p5}
}

@article{KorbanyGullansPiroli2025,
   title = {Long-Range Nonstabilizerness and Phases of Matter},
  author = {Korbany, David Aram and Gullans, Michael J. and Piroli, Lorenzo},
  journal = {Phys. Rev. Lett.},
  volume = {135},
  issue = {16},
  pages = {160404},
  numpages = {7},
  year = {2025},
  month = {Oct},
  publisher = {American Physical Society},
  doi = {10.1103/1hlj-h6t9},
  url = {https://link.aps.org/doi/10.1103/1hlj-h6t9}
}

@article{Veitch2014,
doi = {10.1088/1367-2630/16/1/013009},
url = {https://doi.org/10.1088/1367-2630/16/1/013009},
year = {2014},
month = {jan},
publisher = {IOP Publishing},
volume = {16},
number = {1},
pages = {013009},
author = {Veitch, Victor and Hamed Mousavian, S A and Gottesman, Daniel and Emerson, Joseph},
title = {The resource theory of stabilizer quantum computation},
journal = {New Journal of Physics}
}

@article{LeoneOlivieroHamma2022,
   title = {Stabilizer R\'enyi Entropy},
  author = {Leone, Lorenzo and Oliviero, Salvatore F. E. and Hamma, Alioscia},
  journal = {Phys. Rev. Lett.},
  volume = {128},
  issue = {5},
  pages = {050402},
  numpages = {5},
  year = {2022},
  month = {Feb},
  publisher = {American Physical Society},
  doi = {10.1103/PhysRevLett.128.050402},
  url = {https://link.aps.org/doi/10.1103/PhysRevLett.128.050402}
}

@article{Niroula2024,
   title={Phase transition in magic with random quantum circuits},
  author={Niroula, Pradeep and White, Christopher David and Wang, Qingfeng and Johri, Sonika and Zhu, Daiwei and Monroe, Christopher and Noel, Crystal and Gullans, Michael J},
  journal={Nature physics},
  volume={20},
  number={11},
  pages={1786--1792},
  year={2024},
  publisher={Nature Publishing Group UK London}
}

@article{Bejan2024,
  title={Dynamical magic transitions in monitored Clifford+ T circuits},
  author={Bejan, Mircea and McLauchlan, Campbell and B{\'e}ri, Benjamin},
  journal={PRX Quantum},
  volume={5},
  number={3},
  pages={030332},
  year={2024},
  publisher={APS}
}

@article{Fux2024,
   title={Entanglement--nonstabilizerness separation in hybrid quantum circuits},
  author={Fux, Gerald E and Tirrito, Emanuele and Dalmonte, Marcello and Fazio, Rosario},
  journal={Physical Review Research},
  volume={6},
  number={4},
  pages={L042030},
  year={2024},
  publisher={APS}
}

@article{Suzuki2025,
    title={Quantum complexity phase transitions in monitored random circuits},
  author={Suzuki, Ryotaro and Haferkamp, Jonas and Eisert, Jens and Faist, Philippe},
  journal={Quantum},
  volume={9},
  pages={1627},
  year={2025},
  publisher={Verein zur F{\"o}rderung des Open Access Publizierens in den Quantenwissenschaften}
}

@article{SierantTurkeshi2026,
  title={Theory of magic phase transitions in encoding-decoding circuits},
  author={Sierant, Piotr and Turkeshi, Xhek},
  journal={arXiv preprint arXiv:2603.00235},
  year={2026}
}

@article{Turkeshi2025Spreading,
 title={Magic spreading in random quantum circuits},
  author={Turkeshi, Xhek and Tirrito, Emanuele and Sierant, Piotr},
  journal={Nature Communications},
  volume={16},
  number={1},
  pages={2575},
  year={2025},
  publisher={Nature Publishing Group UK London}
}

@article{Leone2021Chaos,
   title={Quantum chaos is quantum},
  author={Leone, Lorenzo and Oliviero, Salvatore FE and Zhou, You and Hamma, Alioscia},
  journal={Quantum},
  volume={5},
  pages={453},
  year={2021},
  publisher={Verein zur F{\"o}rderung des Open Access Publizierens in den Quantenwissenschaften}
}

@article{Varikuti2026,
 title={Impact of Clifford operations on non-stabilizing power and quantum chaos},
  author={Varikuti, Naga Dileep and Bandyopadhyay, Soumik and Hauke, Philipp},
  journal={Quantum},
  volume={10},
  pages={2017},
  year={2026},
  publisher={Verein zur F{\"o}rderung des Open Access Publizierens in den Quantenwissenschaften}
}

@article{HeinrichGross2019,
 title={Robustness of magic and symmetries of the stabiliser polytope},
  author={Heinrich, Markus and Gross, David},
  journal={Quantum},
  volume={3},
  pages={132},
  year={2019},
  publisher={Verein zur F{\"o}rderung des Open Access Publizierens in den Quantenwissenschaften}
}

@article{Macedo2025,
   title = {Witnessing nonstabilizerness with Bell inequalities},
  author = {Mac\^edo, R. A. and Andriolo, P. and Zamora, S. and Poderini, D. and Chaves, R.},
  journal = {Phys. Rev. A},
  volume = {112},
  issue = {5},
  pages = {L050401},
  numpages = {6},
  year = {2025},
  month = {Nov},
  publisher = {American Physical Society},
  doi = {10.1103/6srg-723m},
  url = {https://link.aps.org/doi/10.1103/6srg-723m}
}

@article{Rains1997,
  title = {A Nonadditive Quantum Code},
  author = {Rains, Eric M. and Hardin, R. H. and Shor, Peter W. and Sloane, N. J. A.},
  journal = {Phys. Rev. Lett.},
  volume = {79},
  issue = {5},
  pages = {953--954},
  numpages = {0},
  year = {1997},
  month = {Aug},
  publisher = {American Physical Society},
  doi = {10.1103/PhysRevLett.79.953},
  url = {https://link.aps.org/doi/10.1103/PhysRevLett.79.953}
}

@article{Cross2009,
  author={Cross, Andrew and Smith, Graeme and Smolin, John A. and Zeng, Bei},
  journal={IEEE Transactions on Information Theory}, 
  title={Codeword Stabilized Quantum Codes}, 
  year={2009},
  volume={55},
  number={1},
  pages={433-438},
  doi={10.1109/TIT.2008.2008136}
}

@article{Chuang2009,
title = "Codeword stabilized quantum codes: Algorithm and structure",
author = "Isaac Chuang and Andrew Cross and Graeme Smith and John Smolin and Bei Zeng",
year = "2009",
doi = "10.1063/1.3086833",
language = "English",
volume = "50",
journal = "Journal of Mathematical Physics",
issn = "0022-2488",
publisher = "American Institute of Physics",
number = "4",
}

@article{YenHsu2009,
  author  = {Yen, Wen-Tai and Hsu, Li-Yi},
  title   = {Optimal Nonadditive Quantum Error-Detecting Code},
  journal = {arXiv preprint arXiv:0901.1353},
  year    = {2009}
}

@article{Baccari2020,
  title = {Scalable Bell Inequalities for Qubit Graph States and Robust Self-Testing},
  author = {Baccari, F. and Augusiak, R. and \ifmmode \check{S}\else \v{S}\fi{}upi\ifmmode \acute{c}\else \'{c}\fi{}, I. and Tura, J. and Ac\'{\i}n, A.},
  journal = {Phys. Rev. Lett.},
  volume = {124},
  issue = {2},
  pages = {020402},
  numpages = {6},
  year = {2020},
  month = {Jan},
  publisher = {American Physical Society},
  doi = {10.1103/PhysRevLett.124.020402},
  url = {https://link.aps.org/doi/10.1103/PhysRevLett.124.020402}
}

\end{document}